\documentclass[a4paper,USenglish,cleveref, autoref, thm-restate]{lipics-v2021}

\usepackage{ifthen}

\newboolean{talk}
\setboolean{talk}{false}
\newboolean{paper}
\setboolean{paper}{false}

\newboolean{LMCSstyle}
\setboolean{LMCSstyle}{false}
\newboolean{IEEEstyle}
\setboolean{IEEEstyle}{false}
\newboolean{lipicsstyle}
\setboolean{lipicsstyle}{true}
\newboolean{eptcsstyle}
\setboolean{eptcsstyle}{false}
\newboolean{sigplanstyle}
\setboolean{sigplanstyle}{false}
\newboolean{PACMPL}
\setboolean{PACMPL}{false}
\newboolean{acmartstyle}
\setboolean{acmartstyle}{false}

\newboolean{needstheorems}
\setboolean{needstheorems}{false}
\newboolean{withimages}
\setboolean{withimages}{false}
\newboolean{withproofs}
\setboolean{withproofs}{true}

\newboolean{french}
\setboolean{french}{false}

\ifthenelse{\boolean{talk}}{}{\usepackage[utf8]{inputenc}}
\ifthenelse{\boolean{PACMPL}}{}{
	\ifthenelse{\boolean{french}}{\usepackage[french]{babel}}{
	  \ifthenelse{\boolean{lipicsstyle}}{}{\usepackage[english]{babel}}}
	  }
\ifthenelse{\boolean{PACMPL}}{}{
	\ifthenelse{\boolean{acmartstyle}}{}{
		\usepackage{amssymb} 
	}
}
\usepackage{amsmath}
\usepackage{graphicx}
\usepackage{bussproofs}
\usepackage{fancybox}
\usepackage{cmll}
\usepackage{stmaryrd}
\usepackage{multirow}
\ifthenelse{\boolean{IEEEstyle}}{
\usepackage[bookmarks={false}]{hyperref}}{
\usepackage{hyperref}}
\usepackage{proof}
\usepackage{hhline}
\usepackage{xspace}
\usepackage{booktabs}
\usepackage{xcolor}
\usepackage{subcaption}
\usepackage{microtype}
\usepackage{wrapfig}
\usepackage{array}
\usepackage{tabularx}
\usepackage{arydshln}
\usepackage{commath}
\ifthenelse{\boolean{IEEEstyle}}{
\setboolean{needstheorems}{true}}{}

\ifthenelse{\boolean{eptcsstyle}}{
\setboolean{needstheorems}{true}}{}

\ifthenelse{\boolean{sigplanstyle}}{
\setboolean{needstheorems}{true}}{}

\ifthenelse{\boolean{needstheorems}}{
\input{\macrospath/thm-environments}}{}

\newcommand{\myproof}[1]{
\ifthenelse{\boolean{withproofs}}{#1}{}}

\newcommand{\withproofs}[1]{
\ifthenelse{\boolean{withproofs}}{#1}{}}

\newcommand{\withoutproofs}[1]{
\ifthenelse{\boolean{withproofs}}{}{#1}}

\newcommand{\red}[1]{{\color{red} #1}}

\newcommand{\tm}{t}
\newcommand{\tmtwo}{u}
\newcommand{\tmthree}{r}
\newcommand{\tmfour}{w}

\newcommand{\var}{x}
\newcommand{\vartwo}{y}

\newcommand{\Rew}[1]{\rightarrow_{#1}}

\renewcommand{\to}{\Rew{}}

\newcommand{\towh}{\Rew{wh}}

\newcommand{\symfont}[1]{\mathsf{#1}}

\newcommand{\varsym}{{\symfont{var}}}

\newcommand{\ctxholep}[1]{[#1]}
\newcommand{\ctxhole}{\ctxholep{\cdot}}

\newcommand{\ctx}{C}
\newcommand{\ctxtwo}{D}

\newcommand{\ctxp}[1]{\ctx\ctxholep{#1}}

\newcommand{\nbvctxtwo}[1]{\nbvctxtwo{#1}}

\newcommand{\defeq}{:=}
\newcommand{\eqdef}{=:}
\newcommand{\grameq}{::=}

\newcommand{\isub}[2]{\{#1/#2\}}

\newcommand{\llbrace}{\{ \kern -0.27em \vert}
\newcommand{\rrbrace}{\vert \kern -0.27em \}}

\newcommand{\grammarpipe}{\mathrel{\big |}}

\renewcommand{\l}{\lambda}
\newcommand{\ie}{\emph{i.e.}\xspace}
\ifthenelse{\boolean{LMCSstyle}}{}{
	\newcommand{\eg}{e.g.\xspace}
	}
\newcommand{\ih}{\textit{i.h.}\xspace}

\newcommand{\blue}[1]{{\color{blue} {#1}}}

\newcommand{\ignore}[1]{}

\newcommand{\myinput}[1]{\ifthenelse{\boolean{withimages}}{\input{#1}}{}}

\newcommand{\nat}{\mathbb{N}}

\newcommand{\size}[1]{|#1|}

\ifthenelse{\boolean{PACMPL}}{\renewcommand{\state}{s}}{
	\ifthenelse{\boolean{acmartstyle}}{\renewcommand{\state}{s}}{
		\newcommand{\state}{s}
	}
}
\newcommand{\statetwo}{{\state'}}
\newcommand{\statethree}{\state''}

\newcounter{numberone}

\newcounter{numbertwo}

\newcommand{\TrPoss}{Logged Positions\xspace}

\newcommand{\Log}{Log\xspace}

\renewcommand{\ctxholep}[1]{\langle #1\rangle}
\newcommand{\ctxtwop}[1]{\ctxtwo\ctxholep{#1}}

\newcommand{\dom}[1]{\mathsf{dom}(#1)}

\newcommand{\tmp}{\tm'}

\newcommand{\reflemma}[1]{Lemma~\ref{l:#1}}

\newcommand{\reflemmaeq}[1]{L.~\ref{l:#1}}

\newcommand{\refprop}[1]{Prop.~\ref{prop:#1}}

\newcommand{\refsect}[1]{Sect.~\ref{sect:#1}}

\newcommand{\refremark}[1]{Remark~\ref{rem:#1}}
\newcommand{\refcoro}[1]{Corollary~\ref{coro:#1}}
\newcommand{\refdef}[1]{Def.~\ref{def:#1}}
\newcommand{\refthm}[1]{Thm.~\ref{thm:#1}}

\newcommand{\reffig}[1]{Fig.~\ref{fig:#1}}

\ifthenelse{\boolean{LMCSstyle}}{	
	
	}{
	
	}

\renewcommand{\isub}[2]{\{#1{\shortleftarrow}#2\}}

\newcommand{\resm}{\psym}
\renewcommand{\resm}{\bullet}

\newcommand{\lpos}{p}
\renewcommand{\lpos}{l}
\newcommand{\lpostwo}{{\lpos'}}

\newcommand{\upp}{\blue{\uparrow}}
\newcommand{\downp}{\red{\downarrow}}
\newcommand{\uppt}{\red{\uparrow}}
\newcommand{\downpt}{\blue{\downarrow}}
\newcommand{\tlog}{L}
\newcommand{\tlogtwo}{\tlog'}

\newcommand{\tape}{T}

\newcommand{\pol}{d}
\newcommand{\poltwo}{\pol'}

\newcommand{\run}{\rho}

\newcommand{\nopolkstate}[6]{(#1,#2,#4,#3,#5,#6)}

\newcommand{\dkstate}[5]{(\red{\underline{#1}},#2,#4,#3,#5)}

\newcommand{\ukstate}[5]{(#1,\blue{\underline{#2}},#4,#3,#5)}

\newcommand{\dstatetab}[4]{\red{\underline{#1}} & #2 & #4 & #3 }

\newcommand{\ustatetab}[4]{#1 & \blue{\underline{#2}} & #4 & #3 }

\newcommand{\ndstatetab}[5]{\red{\underline{#1}} & #2 & #4 & #3 & #5}
\newcommand{\nustatetab}[5]{#1 & \blue{\underline{#2}} & #4 & #3 & #5}

\newcommand{\cons}{{\cdot}}

\newcommand{\IAM}{IAM\xspace}

\newcommand{\JAM}{JAM\xspace}

\newcommand{\PAM}{PAM\xspace}

\newcommand{\SIAM}{SIAM\xspace}
\newcommand{\SJAM}{SJAM\xspace}

\newcommand{\KAM}{KAM\xspace}

\newcommand{\jump}[1]{\mathsf{jump}(#1)}
\newcommand{\depth}[1]{\mathsf{depth}(#1)}

\newcommand{\KJAM}{PaJAM\xspace}
\newcommand{\KSJAM}{SPaJAM\xspace}
\newcommand{\kparam}{k}

\newcommand{\JAMold}{\mathrm{JAM}}

\newcommand{\tomachhole}[1]{\rightarrow_{#1}}
\newcommand{\tomach}{\tomachhole{}}

\newcommand{\btsym}{\mathsf{bt}}
\newcommand{\tomachdotone}{\tomachhole{\resm 1}}
\newcommand{\tomachdottwo}{\tomachhole{\resm 2}}
\newcommand{\tomachvar}{\tomachhole{\varsym}}
\newcommand{\tomachbttwo}{\tomachhole{\btsym 2}}

\newcommand{\iamdap}{\tomachdotone}
\newcommand{\iamdlamone}{\tomachdottwo}
\newcommand{\iamdvar}{\tomachvar}

\newcommand{\argsym}{\mathsf{arg}}
\newcommand{\tomachdotthree}{\tomachhole{\resm 3}}
\newcommand{\tomachdotfour}{\tomachhole{\resm 4}}
\newcommand{\tomacharg}{\tomachhole{\argsym}}
\newcommand{\tomachbtone}{\tomachhole{\btsym 1}}
\newcommand{\iamuapltwo}{\tomachdotthree}
\newcommand{\iamulam}{\tomachdotfour}
\newcommand{\iamuaplone}{\tomacharg}

\newcommand{\jumpsym}{\mathsf{jmp}}
\newcommand{\tomachjump}{\tomachhole{\jumpsym}}
\newcommand{\iamujump}{\tomachjump}

\newcommand{\tosiam}{\rightarrow_{\textsc{SIAM}}}

\newcommand{\toksjam}{\rightarrow_{\mathrm{\KSJAM}}}

\newcommand{\tokjam}{\rightarrow_{\mathrm{\KJAM}}}

\newcommand{\stempty}{\epsilon}

\newcommand{\la}[1]{\lambda #1.}

\newcommand{\exstates}{\mathcal{E}}

\newcommand{\midd}{\; \; \mbox{\Large{$\mid$}}\;\;}

\newcommand{\bigo}[1]{\mathcal{O}(#1)}

\newcommand{\cbn}{CbN\xspace}
\newcommand{\ccbn}{Closed \cbn}

\newcommand{\mset}[1]{[#1]}
\newcommand{\emmset}{[\,]}

\newcommand{\initty}{\star}

\newcommand{\linty}{A}
\newcommand{\lintytwo}{\linty'}

\newcommand{\lintyb}{B}

\newcommand{\ltyctx}{\mathbb{\linty}}
\newcommand{\ltyctxp}[1]{\ltyctx\ctxholep{#1}}
\newcommand{\ltyctxtwo}{\ltyctx'}
\newcommand{\ltyctxtwop}[1]{\ltyctxtwo\ctxholep{#1}}
\newcommand{\ltyctxthree}{\ltyctx''}
\newcommand{\ltyctxthreep}[1]{\ltyctxthree\ctxholep{#1}}

\newcommand{\ltyctxb}{\mathbb{\lintyb}}
\newcommand{\ltyctxbp}[1]{\ltyctxb\ctxholep{#1}}
\newcommand{\ltyctxtwob}{\ltyctxb'}
\newcommand{\ltyctxtwobp}[1]{\ltyctxtwob\ctxholep{#1}}
\newcommand{\ltyctxthreeb}{\ltyctxb''}
\newcommand{\ltyctxthreebp}[1]{\ltyctxthreeb\ctxholep{#1}}

\newcommand{\mty}{\mathcal{A}}
\newcommand{\mtytwo}{\mathcal{B}}
\renewcommand{\mty}{S}
\renewcommand{\mtytwo}{\mty'}
\newcommand{\arr}[2]{#1\rightarrow #2}

\newcommand{\mtyctx}{\mathbb{\mty}}

\newcommand{\tye}{\Gamma}
\newcommand{\tyetwo}{\Delta}

\newcommand{\tjudg}[3]{#1\vdash #2:#3}

\newcommand{\tjudgi}[3]{#1\vdash_i #2:#3}

\newcommand{\tjudgw}[4]{#1\vdash^{\textcolor{violet}{#2}} #3:#4}

\newcommand{\tyvar}{\textsc{T-var}}
\newcommand{\tylamstar}{\tylam_\star}
\newcommand{\tylam}{\textsc{T-}\lambda}
\newcommand{\tyapp}{\textsc{T-@}}

\newcommand{\tyd}{\pi}

\newcommand{\pof}{\;\triangleright}

\newcommand{\WeightTimekJAM}[2]{\mathbf{W}^{#2}_{\mathrm{\KJAM}}(#1)}

\newcommand{\occstar}[1]{\norm{#1}}

\newcommand{\DiPref}[1]{\mathsf{DiPref}(#1)}
\newcommand{\myldots}{...}

\newcommand{\ruleoc}{J}

\newcommand{\extsym}{\symfont{ext}}
\newcommand{\bisimtypes}{\simeq_{\extsym}}
\newcommand{\etape}[1]{\tape_{\extsym}(#1)}
\newcommand{\etapeaux}[2]{\tape_{\extsym}^{#2}(#1)}
\newcommand{\etapeauxs}[1]{\tape_{\extsym}^{\state}(#1)}
\newcommand{\elog}[1]{\tlog_{\extsym}(#1)}
\newcommand{\elpos}[1]{\lpos_{\extsym}(#1)}
\newcommand{\estate}[1]{\state_{\extsym}(#1)}

\newcommand{\focus}{f}
\newcommand\mydots{\hbox to .6em{.\hss.}}

\renewcommand{\state}{q}

\newcommand{\sizeto}[1]{|{#1}|_{\rightarrow}}

\newcommand{\sizemult}[1]{|{#1}|_{[]}}
\newcommand{\maxto}[1]{\#_{\rightarrow}({#1})}
\newcommand{\maxmult}[1]{\#_{[]}({#1})}

\usepackage{thmtools}
\usepackage{thm-restate}

\hideLIPIcs

\title{On Jumps, Interactions, and Intersection Types} 

\author{Stefano Catozi}{Université Sorbonne Paris Nord, Villetaneuse, France }{catozi@lipn.univ-paris13.fr}{https://orcid.org/0009-0007-0829-5789}{}

\author{Ugo Dal Lago}{Università di Bologna, Italia and Inria, France }{ugo.dallago@unibo.it}{https://orcid.org/0000-0001-9200-070X}{}

\author{Gabriele Vanoni}{IRIF, Université Paris Cité, France }{gabriele.vanoni@irif.fr}{https://orcid.org/0000-0001-8762-8674}{}

\authorrunning{S. Catozi, U. Dal Lago and G. Vanoni} 



\category{} 

\relatedversion{} 

\funding{The third author is supported by the European Union's Horizon 2020 research and innovation programme under the Marie Sklodowska-Curie grant agreement No.~101034255.}

\nolinenumbers 

\begin{document}

\maketitle

\begin{abstract}
    The Jumping Abstract Machine (JAM), an evaluation mechanism for the
    $\lambda$-calculus, was introduced by Danos and Regnier 
    as an optimization of the Interaction Abstract
    Machine (IAM), itself an operational counterpart to Girard's Geometry of
    Interaction and Abramsky \emph{et al}. game semantics. Moreover, the JAM is isomorphic to the  Pointer Abstract Machine (PAM), the syntactical counterpart of Hyland and Ong's game semantics. We study a generalization of the JAM, that we call the
    Parametric Jumping Abstract Machine (PaJAM) and show that there
    is a tight correspondence between the PaJAM and non-idempotent
    intersection types: given a normalizing term $t$, the number of steps 
    taken by the PaJAM when evaluating $t$ can be extracted from
    its non-idempotent intersection type derivation. Remarkably, fixing the
    backtracking depth of the PaJAM, one can easily recover
    both the JAM/PAM, when the depth is constrained to be zero, and the IAM, when
    it is instead unconstrained. Exploiting type-theoretic machinery, we analyze the complexity of the PaJAM, showing that it is \emph{polynomial} in the number of weak head $\beta$ steps, giving rise to a \emph{reasonable} cost model, for each \emph{finite} bound on the backtracking depth.
\end{abstract}

\section{Introduction}

This paper deals with two central concepts in the theory of the $\lambda$-calculus, namely \emph{abstract machines}~\cite{landin_mechanical_1964,DBLP:journals/tcs/Plotkin75,DBLP:conf/icfp/AccattoliBM14} on the one hand and \emph{intersection types}~\cite{CoppoD80,BucciarelliKV17,DBLP:conf/lics/BonoD20} on the other. We are particularly interested in the relationship between these two concepts and our main contribution is a new generalized tight correspondence result. This will allow us to analyze the efficiency of a family of machines via type-theoretic tools. We give a bit of context before delving into the technical details.

\subparagraph{Abstract Machines.} An abstract machine is an automaton designed to evaluate a term of the 
$\lambda$-calculus (or of similar calculi), \ie to compute a representation of 
its normal form. In doing so, the abstract machine usually 
conforms to a notion of reduction, e.g., \emph{call-by-value} or 
\emph{call-by-name} reduction~\cite{DBLP:journals/tcs/Plotkin75}. Abstract machines, unlike mere rewriting, are 
thought of as (relatively) low-level, first-order, descriptions of the evaluation process, 
such that each individual step is in some sense \emph{elementary}. As it is well 
known, there are a variety of abstract machines, each with its own correctness 
and efficiency properties. Among the best-studied machines, there is the so-called 
Krivine Abstract Machine (KAM)~\cite{krivine_call-by-name_2007}, which implements call-by-name (a.k.a. weak-head) reduction. Machines which are similar in spirit to the Krivine machine but implementing 
call-by-value or call-by-need reduction have also been introduced~\cite{DBLP:conf/icfp/AccattoliBM14,DBLP:conf/ppdp/AccattoliB17}.

\subparagraph{Game Machines.} Since the pioneering work of Danos, Herbelin and Regnier~\cite{DBLP:conf/lics/DanosHR96}, it is well 
known that there are machines implementing weak-head reduction like the 
KAM, but which have an even lower-level behavior, closely related to game 
semantics and Girard's geometry of interaction (GoI)~\cite{girard_geometry_1989}. In particular, the 
so-called Interaction Abstract Machine (IAM)~\cite{mackie_geometry_1995,DR99,PPDP2020} can be seen 
as an automata-theoretic counterpart of the GoI or of AJM games~\cite{DBLP:journals/iandc/AbramskyJM00}, and evaluates terms
through a purely \emph{local} and \emph{reversible} process. There is also the so-called 
Pointer Abstract Machine (PAM)~\cite{DBLP:conf/lics/DanosHR96,Danos04headlinear}, which is instead inspired, as its name says, 
by the mechanism of justification pointers from Hyland and Ong's games~\cite{DBLP:journals/iandc/HylandO00}.  As discovered by Danos and
Regnier~\cite{DR99}, the PAM is isomorphic to the Jumping Abstract Machine (JAM), an optimization of the IAM obtained by allowing the latter to have, in certain circumstances, the possibility to \emph{jump} from one subterm to another, thus losing locality and reversibility.
These machines have been studied with respect to their correctness and 
relative performance. In particular, Accattoli, Dal Lago and Vanoni, have 
recently shown that the KAM can be seen as an improvement to the JAM (and thus of the PAM), and that the latter is an improvement to the IAM~\cite{POPL2021,TCS2026}.  Noticeably, while switching from the KAM to the JAM involves a \emph{polynomial} overhead, the IAM can be \emph{exponentially} less efficient than the JAM.

\subparagraph{Intersection Types and the KAM.} What role do intersection types play in all this? In their non-idempotent variant~\cite{DBLP:conf/tacs/Gardner94,DBLP:journals/logcom/Kfoury00}, intersection types are known to be a syntactic counterpart of the \emph{relational model} of linear logic~\cite{DBLP:journals/apal/BucciarelliE01}. Starting from de Carvalho's results~\cite{deCarvalho18}, the size of the type derivation for a term $\tm$ has been put in relation to the number of reduction steps needed to evaluate the term $\tm$ using the KAM. The correspondence can be made tight by considering so-called \emph{tight typings}~\cite{DBLP:journals/jfp/AccattoliGK20} or by introducing a special type $\initty$ corresponding to normal forms. In such correspondences, each KAM step is put in relation to the application of a typing rule, and this relation is bijective: subterms that are used many times appear many times in the derivation, while terms that are \emph{not} used simply are not subject to typing.

\subparagraph{Intersection Types and Game Machines.} Accattoli \emph{et al.} recently showed that the correspondence just 
mentioned, surprisingly, scales to the IAM~\cite{POPL2021}. Remarkably, this does not require 
any modification to the type system itself, but only amounts to altering the 
way each typing rule is weighted. While in the KAM each rule occurrence counts as 
\emph{one}, the IAM requires assigning a weight equal to \emph{the size of the conclusion
type}. In other words, the number of IAM steps is precisely reflected by the 
number of occurrences of the type $\star$ to the right of $\vdash$ in the 
corresponding type derivation. Since types can be exponentially bigger than terms, this explains the aforementioned efficiency gap.
But what about the JAM or the PAM? Would it be possible to precisely characterize the 
dynamic behavior of the JAM  by means of non-idempotent intersection types? More fundamentally: given that there is a phase transition between the JAM and the IAM, what is it that \emph{drives} this transition? And why, by contrast, do the KAM and the JAM turn out to have both a polynomial overhead? These are precisely the questions that this work addresses. And the results we 
obtain confirm that non-idempotent intersection types are indeed a very flexible 
tool. Once again, and surprisingly, it is only a matter of modifying the weight given to each 
rule. If in the IAM we count \emph{all} occurrences of $\star$ in the 
conclusion of each rule, the JAM requires to count only those occurrences of 
$\star$ that are at \emph{a depth equal to $0$ or $1$}. As a result, the blowup can no longer be exponential and becomes polynomial instead. In other words, the \JAM/\PAM is much closer to the \KAM than to the \IAM, even if it was designed as an optimization of the latter. The situation, a bit more formally, is similar to the one in the figure below, where $n$ is the total size of the types occurring in the derivation.
	\begin{center}
		\includegraphics{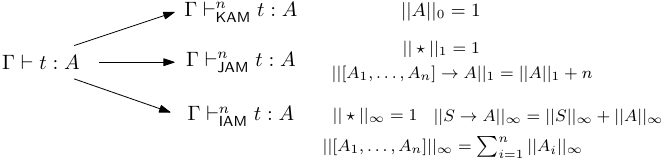}
	\end{center}
As we explain below, however, we do much more, showing that the IAM and the JAM are both instances of a parametric version of the latter. Indeed, looking at the figure above, it is natural to wonder whether there is anything \emph{in between} the JAM and the IAM in which some (but not all) occurrences of $\star$ at depth higher than $1$ are counted. Moreover, are these intermediate machines efficient?

\subparagraph{Contributions.} 
This work introduces an abstract machine, called the Parametric Jumping Abstract Machine (PaJAM), obtained by generalizing Danos and Regnier’s jumps so that they \emph{do not} occur until a certain backtracking depth $k$. Indeed, the intuition is that the \IAM computes using a possibly nested backtracking mechanism and jumping allows one to skip this backtracking phase. However, because of the nesting, one could be inside $k$ backtracking phases and decide at some point not to go one level deeper, and instead jump, thus remaining at depth $k$. More specifically, we develop the following results:
\begin{itemize}
	\item We define a generalization of the JAM, the \KJAM, which is a parametric machine, in \refsect{k-machine}. The \JAM can be recovered by setting the backtracking depth $k$ as $0$, while the IAM can be seen as a degenerate version of the PaJAM where $k = \infty$.
	\item As shown in \refsect{measures}, the number of steps performed by the PaJAM when evaluating a term $t$ can be measured, following the spirit of figure above, by letting the size of a type depend on $k$ and counting occurrences of the base type $\star$ up to depth $k$:
	
	\begin{center}
	\includegraphics{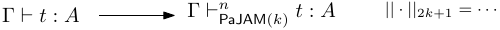}
	\end{center}
	\item These results are then used in \refsect{complexity} to show that every instance of the PaJAM with $k<\infty$ simulates $\beta$ reduction with a polynomial overhead, thus being time invariant.
\end{itemize}
 There are two take-home messages: not only non-idempotent intersection types 
allow a fine-grained and direct analysis of the dynamic properties of reduction, 
even when the latter is performed by abstract machines, but the 
performances of the JAM (and thus of the PAM) and the IAM can be studied within a \emph{single} framework, which also gives a very simple account of their relative 
performance. 

\section{The Closed Call-by-Name \texorpdfstring{$\l$}{Lambda}-Calculus}

Let $\mathcal{V}$ be a countable set of \emph{variables}. The \emph{terms}
of the $\lambda$-calculus are defined as follows:
\[ \begin{array}{rcl}
		\tm,\tmtwo,\tmthree,
		\tmfour             & \grameq              &
		x\in\mathcal{V}\midd \lambda x.\tm\midd \tm\tmtwo
	\end{array} \]
\emph{Free} and \emph{bound variables} are defined as usual:
$\la\var\tm$ binds $\var$ in $\tm$. Terms are considered modulo
$\alpha$-equivalence, and $\tm\isub\var\tmtwo$ denotes the
capture-avoiding (meta-level) \emph{substitution} of all free
occurrences of $\var$ in $\tm$ for $\tmtwo$.
\emph{Contexts} are $\lambda$-terms containing exactly 
one occurrence of a special symbol, the hole $\ctxhole$, intuitively standing for a removed subterm. Here we adopt 
\emph{levelled contexts}, whose index, \ie\ the level, stands for the number of 
arguments (\ie\ the 
number of !-boxes in linear logic terminology) the hole lies in:
\begin{center}$
	\begin{array}{rclrrrcl}
		\ctx_0		& \grameq &	\ctxhole \midd \la\var\ctx_0 \midd \ctx_0\tm\ ;
		&&&
		\ctx_{n+1}	& \grameq &	\ctx_{n+1}\tm\midd\la\var\ctx_{n+1}\midd\tm\ctx_{n}.
	\end{array}
	$\end{center}
We simply write $\ctx$ for a context whenever the level is not relevant. 
The operation replacing the hole $\ctxhole$ with a term $\tm$ 
in a context $\ctx$ is noted $\ctxp\tm$ and called \emph{plugging}. This operation can potentially capture free variables.
We consider a \emph{call-by-name} (CbN) operational semantics which we restrict to \emph{closed} terms, \ie without any free variable. This is defined as follows:
\[
	(\la\var \tm) \tmtwo \tmthree_1 \ldots
	\tmthree_h \ \ \towh \ \ \tm \isub\var \tmtwo
	\tmthree_1 \ldots \tmthree_h \qquad\qquad h\geq 0.
\]
This way, $\lambda$-abstractions are the only normal forms.

For any notion of reduction $\to$, we denote by $\to^*$ the reflexive and transitive closure of $\to$, and we
define $\to^n$ by induction as $\to^0 \defeq \mathsf{id}$ (the identity
relation), and $\to^{n+1}\defeq \to\,\to^n$.

\section{The Parametric JAM}
\label{sect:k-machine}
The Jumping Abstract Machine ($\JAMold$) is introduced in \cite{DR99} as an optimization of the interaction abstract machine (\IAM)~\cite{mackie_geometry_1995,DR99,PPDP2020}. The difference between them lies in the fact that a complex, and possibly nested, \emph{backtracking} mechanism, which is the peculiarity of the \IAM, is short-circuited in the \JAM, making it more efficient.
In this section, we introduce a new machine, the Parametric \JAM (or \KJAM), which is a machine that behaves like the \IAM until a fixed backtracking nesting depth (specified by the \emph{initial parameter} of the machine), and then proceeds, avoiding backtracking, as the \JAM.

\begin{figure}[t]
	\input{machines/LJAM-strat-data}
	\\[3pt]
	\scalebox{0.97}{\input{machines/LJAM-strat}}
	\caption{The data structures and transitions of the Parametric \JAM (\KJAM).}
	\label{fig:LJAM-strat}  
\end{figure}

Intuitively, the behaviour of the \KJAM, generalizing that of the \IAM and the \JAM, can be seen as that of a token that travels around the syntax tree of the program under evaluation, possibly making some \emph{jumps}. The evaluation is done in two phases. In one, noted $\downp$, the machine looks for the head variable, while in the other one, noted $\upp$, the machine looks for the argument which should replace it. We redirect the reader to~\cite{PPDP2020,POPL2021} for detailed explanations on the behavior of this family of machines.

\subparagraph{\KJAM States.}
The transitions of the \KJAM and all the data structures are defined in \reffig{LJAM-strat}. As we can see from the definition, the machine relies on two mutually inductive data structures, namely \emph{logs} and \emph{logged positions}.
We note with $\tlog$ a log of arbitrary length.
In order to implement the optimization mechanism introduced earlier, each machine state carries an integer parameter $k \in\nat\cup \{\infty\}$, the \emph{backtracking (nesting) depth}, indicating how many backtracking levels the machine may still enter.
Initial states have the form $\state_{\tm}^{\kparam} \defeq (\red{\tm},\ctxhole,\epsilon,\epsilon,\kparam)$, with $\kparam$ being the initial parameter of the machine.
Directions are often omitted and represented via colors and underlining: 
$\downp$ is represented by a
\red{red} and underlined code term, $\upp$ by a \blue{blue} and underlined code context.

\subparagraph{Transitions.} The transitions of the \KJAM are in \reffig{LJAM-strat} and their union is noted $\tokjam$. Note that the machine is deterministic, but not reversible. A \emph{run} $\run: \state \tomach^*\statetwo$ is a possibly empty sequence of transitions, whose length is noted 
$\size\run$. If it starts from an initial state it is called \emph{initial}, and if moreover it ends on a final state it is called \emph{complete}.
Consider an execution of the \KJAM starting from the initial state $\state_{\tm}^\kparam$, only two cases can happen: either the execution diverges or it terminates in a state $\state$ of the form 
$\state=\dkstate{\la\var\tmtwo}{\ctx}{\epsilon}{\tlog}{\kparam}$. We call such states \emph{final}. 
The fact that no other shapes are possible can proved as in~\cite{POPL2021}.
Notice how the parameter $k$ is decremented by 1 after entering a backtracking phase via a $\tomachbtone$ transition and is incremented by 1 after exiting a backtracking with a $\tomachbttwo$ transition. When $k$ reaches $0$, the $\tomachbtone$ transition is replaced by a $\tomachjump$, allowing the machine to avoid the backtracking runs after a certain depth.

\begin{remark}
It is clear that if we start from an initial state $\state_\tm^0$, the \KJAM run never performs backtracking and thus behaves exactly like the \JAM~\cite{POPL2021}.
On the other hand, when the initial parameter is set to $\infty$, no jump is ever performed, the \KJAM thus behaving like the \IAM of~\cite{PPDP2020}.
\end{remark}
The \IAM is well-known to be correct for Closed CbN. This property is extended to the \JAM in \cite{POPL2021}, using a simulation argument.
The proof smoothly adapts to the \KJAM, and this is why we do not detail it here. 
\begin{proposition}[Correctness]
	\label{prop:implements-cbn}
	$\state_\tm^k$ terminates if and only if $\tm$ has normal form.
\end{proposition}

The direction of reachable states is directly related to the backtracking nesting depth, as the next invariant shows.
Let $|T|_l$ be the number of logged position in a tape $T$.
\begin{proposition}[Phase Invariant]
	\label{prop:tape-logs-jam-strat}
	Let $s=(t,C,L,T,d,k')$ be a $\mathrm{\KJAM}$ state reachable from $\state_u^{k<\infty}$. If $d=\uppt$ then $|T|_l = 2(\kparam-k')$ and if $d=\downpt$ then $|T|_l = 2(\kparam-k')+1$. 
\end{proposition}


\subparagraph{An Example.}
\label{para:ex-sjam}
We give an example of execution of the \KJAM with parameter $\kparam=1$. The first part of the execution looks for the head variable of the term $(\lambda x.xx)(\lambda y.y)$, namely $x$.
\[{\footnotesize
	\begin{array}{l|c|c|c|c|c|c}
		&\mathsf{Sub}\mbox{-}\mathsf{term} & \mathsf{Context} & \mathsf{\Log} & 
		\mathsf{Tape} & \mathsf{Dir} & \mathsf{Depth}
		\\
		\cline{1-7}
		&\ndstatetab{(\la x xx)(\la y y)} {\ctxhole} {\epsilon} {\epsilon} 
		\downp & 1\\
		\iamdap&\ndstatetab{\la x xx} {\ctxhole(\la y y)} {\resm} {\epsilon} 
		\downp& 1\\
		\iamdlamone&\ndstatetab{xx} {(\la x \ctxhole)(\la y y)} {\epsilon} 
		{\epsilon} \downp& 1\\
		\iamdap&\ndstatetab{x} {(\la x \ctxhole x)(\la y y)} {\resm} {\epsilon} 
		\downp& 1\\
		\iamdvar&\nustatetab{\la x xx} {\ctxhole(\la y y)} 
		{(x,(\la x \ctxhole x)(\la\vartwo\vartwo),\epsilon)\cdot\resm} 
		{\epsilon}{\upp}& 1 \\
\end{array}}
\]
Once the variable head $x$ is found, the machine switches to $\upp$ and starts to search for its argument, while saving the logged position of $x$ on the tape. We set $\lpos_\var\defeq(x,(\la x \ctxhole x)(\la\vartwo\vartwo),\epsilon)$.
\[{\footnotesize
	\begin{array}{l|c|c|c|c|c|c}
		&\mathsf{Sub}\mbox{-}\mathsf{term} & \mathsf{Context} & \mathsf{\Log} & 
		\mathsf{Tape} & \mathsf{Dir}& \mathsf{Depth}
		\\
		\cline{1-7}
		&\nustatetab{\la x xx} {\ctxhole(\la y y)} 
		{\lpos_\var\cdot\resm} 
		{\epsilon}{\upp}& 1 \\
		\iamuaplone&\ndstatetab{\la y y} {(\la x xx)\ctxhole} {\resm} 
		{\lpos_\var}\downp& 1 \\
		\iamdlamone&\ndstatetab{y} {(\la x xx)(\la y \ctxhole)} {\epsilon} 
		{\lpos_\var}\downp& 1 \\
		\iamdvar&\nustatetab{\la y y} {(\la x xx)\ctxhole} 
		{(y,(\la\var \var\var)(\la\vartwo\ctxhole),\lpos_\var)} {\lpos_\var}\upp& 1 \\
\end{array}}
\]
The $\tomacharg$ transition saves the position $l_x$ on the log and the machine starts looking for the head variable of the argument. Having found it, it saves its position on the tape. Now the machine starts looking for the argument of $\lambda y.y$ in $\upp$ mode.  However, $\lambda y.y$ was not the left side of an application forming a $\beta$-redex. Indeed, it was virtually substituted for the first occurrence of $x$, thus creating
the ``virtual'' redex $(\lambda y.y)x$. Its argument is thus the second occurrence of $x$. Since the backtracking parameter is $1$, the machine is allowed to perform one level of backtracking. It therefore starts a backtracking run with a $\tomachbtone$ transition, in order to find the variable associated with the logged position $l_x$ stored in the log. We set $\lpos_\vartwo\defeq (y,(\la\var\var)(\la\vartwo\ctxhole),\lpos_\var)$.\\
\[{\footnotesize
	\begin{array}{l|c|c|c|c|c|c}
		&\mathsf{Sub}\mbox{-}\mathsf{term} & \mathsf{Context} & \mathsf{\Log} & 
		\mathsf{Tape} & \mathsf{Dir}&\mathsf{Depth}
		\\
		\cline{1-7}
		&\nustatetab{\la y y} {(\la x xx)\ctxhole} 
		{\lpos_\vartwo} {(x,(\la x \ctxhole x)(\la\vartwo\vartwo),\epsilon)}\upp& 1 \\
		\tomachbtone&\ndstatetab{\la x xx} { \ctxhole(\la y y)} 
		{(x,(\la x \ctxhole x)(\la\vartwo\vartwo),\epsilon)\cdot\lpos_\vartwo} {\epsilon}\downp& 0 \\
		\tomachbttwo&\nustatetab{x} {(\la x \ctxhole x)(\la y y)} 
		{\lpos_\vartwo} {\epsilon}\upp& 1 \\
\end{array}}\]
If instead we had started with an initial parameter equal to $0$, the machine would have skipped the backtracking run between $\tomachbtone$ and $\tomachbttwo$, and would have jumped directly to the position stored in $l_x$, as shown below.
\[{\footnotesize
	\begin{array}{l|c|c|c|c|c|c}
		&\mathsf{Sub}\mbox{-}\mathsf{term} & \mathsf{Context} & \mathsf{\Log} & 
		\mathsf{Tape} & \mathsf{Dir}&\mathsf{Depth}
		\\
		\cline{1-7}
		&\nustatetab{\la y y} {(\la x xx)\ctxhole} 
		{\lpos_\vartwo} {(x,(\la x \ctxhole x)(\la\vartwo\vartwo),\epsilon)}\upp& 0 \\
		\iamujump&\nustatetab{x} {(\la x \ctxhole x)(\la y y)} 
		{\lpos_\vartwo} {\epsilon}\upp& 0 \\
\end{array}}
\]
Notice that the states reached after $\tomachbttwo$ and $\tomachjump$ are the same, except for the value of the parameter.
Continuing the original run, the machine can now find immediately the argument, being the second occurrence of $x$. Then, it starts looking for the argument of this variable, finding of course a new copy of $\lambda y.y$. 
\[{\footnotesize
	\begin{array}{l|c|c|c|c|c|c}
		&\mathsf{Sub}\mbox{-}\mathsf{term} & \mathsf{Context} & \mathsf{\Log} & 
		\mathsf{Tape} & \mathsf{Dir} & \mathsf{Depth}
		\\
		\cline{1-7}
		&\nustatetab{x} {(\la x \ctxhole x)(\la y y)} 
		{\lpos_\vartwo} {\epsilon}\upp& 1 \\
		\iamuaplone&\ndstatetab{x} {(\la x x\ctxhole)(\la y y)} {\epsilon} 
		{\lpos_\vartwo}\downp& 1 \\
		\iamdvar&\nustatetab{\la x xx} {\ctxhole(\la y y)} 
		{(x,(\la x x\ctxhole)(\la\vartwo\vartwo),\lpos_\vartwo)} {\epsilon}\upp & 1
		\\
		\iamuaplone&\ndstatetab{\la y y} {(\la x xx)\ctxhole} {\epsilon} 
		{(x,(\la x x\ctxhole)(\la\vartwo\vartwo),\lpos_\vartwo)}\downp& 1\\
\end{array}}\]
The computation then stops, signaling that the term has a normal form.

\section{The Sequence \KJAM}
In this section, we present the definition of the Sequence Parametric \JAM (\KSJAM), which is the type theoretic version of the \KJAM and which is strongly bisimilar to it. Instead of being defined in an automata-theoretic way, the \KSJAM is defined as a machine \emph{acting on} the intersection type derivation for the term under evaluation. Indeed, abstracting from the concrete data structures, as done in~~\cite{POPL2021,LICS2021}, allows us to better reason about the complexity of the machine. We do not have here the space to explain in detail how the type system works, and thus we refer to~\cite{BucciarelliKV17,DBLP:conf/lics/BonoD20}.

\subparagraph{Sequence Types.} We define sequence types, or non-idempotent and not commutative intersection types, as follows:
\begin{figure}[t]
	\[
	\begin{array}{c@{\hspace{1cm}}c@{\hspace{1cm}}c}
		\infer[\tyvar]{\tjudgw{\var:\mset{\linty}}{\occstar{\linty}_k}{\var}{\linty}}{}
		&
		\infer[\tylam]{\tjudgw{\tye}{w+\occstar{\arr\mty\linty}_k}{\lambda\var.\tm}{\arr{\mty}{\linty}}} 
		{\tjudgw{\tye,\var:\mty}{w}{\tm}{\linty}}
		&
		\infer[\tylamstar]{\tjudgw{}{0}{\lambda\var.\tm}{\initty}}{}
		\\[8pt]
		\multicolumn{3}{c}{\infer[\tyapp]{\tjudgw{\tye\uplus_i 
					 \tyetwo_i  }{w+\sum_i 
					v_i+\occstar{\linty}_k}{\tm\tmtwo}{\linty 
			}}{\tjudgw{\tye}{w}{\tm}{\arr{\mset{\lintytwo_1,\ldots,\lintytwo_n}}{\linty}}
				& 
				\mset{\tjudgw{\tyetwo_i}{v_i}{\tmtwo}{\lintytwo_i}}_{i\in\mset{1,\ldots,n}}}}
	\end{array}
	\]
	\vspace{-8pt}
	\caption{The weighted sequence type system.}
	\label{fig:mtypesystem}
\end{figure}
\begin{center}$
	\begin{array}{rrcl@{\hspace{1cm}}rrcl}
		\textsc{Linear 
			types}&\linty,\lintytwo&\grameq&\initty\grammarpipe\arr{\mty}{\linty} \\[3pt]
		\textsc{Sequence types}&\mty,\mtytwo&\grameq&\mset{\linty_1,\ldots,\linty_n}\;\;\;\; n\geq 0
	\end{array}$
\end{center}
Please note that there is only one ground type $\initty$. This type
will be used to type normal forms, and these are precisely abstractions
in \ccbn. The empty
sequence is denoted by $\emmset$ and is used to type subterms which will
be discarded along the evaluation. The concatenation of two sequences $\mty$
and $\mtytwo$ is denoted by $\mty\uplus\mtytwo$. Please note that the
type system is \emph{linear}. This means that weakening and contraction
are not allowed. 
Type judgments have the form $\tjudg{\tye}{\tm}{\linty}$, where $\tye$ is a type
environment, \ie a total function from variables to sequence types that
maps all but finitely many variables to $\emmset$. We write $\tye =
\var_1:\mty_1,\ldots,\var_n:\mty_n$ if $\dom\tye = \set{\var_1,\ldots,\var_n}$,
where $\dom{\tye}$ is defined as the (finite) set of variables in $\tye$ that
are mapped to non-empty sequence types. The union of two type environments
$\tye,\tyetwo$ is denoted by $\tye\uplus\tyetwo$ and stands for the type
environment mapping every variable $\var \in \dom{\tye} \cup \dom{\tyetwo}$ to
the sequence $\tye(\var)\uplus\tyetwo(\var)$. The typing rules are in Figure~\ref{fig:mtypesystem}, please ignore the weights for now.
Type derivations are denoted by $\tyd$, and we write $\tyd \pof
\tjudg{\tye}{\tm}{\linty}$ to denote a type derivation $\tyd$ ending in the
judgment $\tjudg{\tye}{\tm}{\linty}$.

\subparagraph{The \KSJAM.}
\label{para:sjam}

The idea behind the \KSJAM is simple: the machine moves around a fixed type derivation $\tyd\pof\tjudg{}{\tm}{\initty}$, to be thought of as the 
code. The \emph{focus} of the machine is expressed as an occurrence of a type judgment $\ruleoc$ 
of $\tyd$. Like the \KJAM, the \KSJAM has two possible directions, noted $\downpt$ and $\uppt$.\footnote{Type 
	derivations are upside-down w.r.t. to the term structure, then direction $\downp$ of the \KJAM becomes here $\uppt$, and 
	$\upp$ is $\downpt$.} In direction $\uppt$ the machine looks at the rule above the focused judgment, in direction 
$\downpt$ at the rule below. Each state contains a type context
$\ltyctx$ isolating an occurrence of 
$\initty$ in the type $\linty$ of the focused judgment (occurrence) 
$\tjudg{\tye}{\tmtwo}{\linty}$, defined as follows:
\begin{center}$
	\begin{array}{rrcll@{\hspace{.3cm}}rrcl}
		\textsc{Type ctxs}&\ltyctx & \grameq  &\ctxhole \grammarpipe \arr\mty\ltyctx 
		\grammarpipe 
		\arr{\mtyctx}\linty\\
		
		\textsc{Sequence ctxs}&\mtyctx&\grameq&\mset{\linty_1,.., 
			\ltyctx,..,\linty_n} \quad n\geq 0
	\end{array}$
\end{center}
Moreover, as in the \KJAM, each state carries a parameter $k\in\nat\cup\{\infty\}$.
Summing up, a state $\state$ is a 
quintuple $(\tyd, \ruleoc, \ltyctx, d,k)$. If $\ruleoc$ is in the form 
$\tjudg{\tye}{\tmtwo}{\linty}$, we often write $\state$ as 
$\tjudg{}{\tmtwo}{\ltyctxp{\initty^k_\pol}}$, where $\ltyctxp{\initty}=\linty$, in fact type environments play no role here. 
Given a derivation $\pi\pof\vdash t:\star$, the initial state of the machine will be $(\tyd, \ruleoc, \ctxhole, \uppt,\kparam)$, with $\ruleoc$ the final judgment of $\pi$ and $k\in\nat\cup\{\infty\}$ the initial parameter of the machine.
\begin{figure}[t]
	\scalebox{.93}{
\begin{tabular}{c}
{\footnotesize$\begin{array}{ccc||ccc}
	\infer{\tjudg{}{\red{\tm\tmtwo}}{\ltyctxp{\initty^k_\uppt} (=\linty)}} 
	{\tjudg{}{\tm}{\arr{\mty}{\linty}} & \mset{\vdash}} 
	&\tomachdotone&
	\infer{\tjudg{}{\tm\tmtwo}{\linty 
		}}{\tjudg{}{\red\tm}{\arr{\mty}{\ltyctxp{\initty^k_\uppt}}} & 
		\mset{\vdash}}
		&
	\infer{\tjudg{}{\red{\lambda\var.\tm}}{\arr{\mty} 
			{\ltyctxp{\initty^k_\uppt}}}} 
	{\tjudg{}{\tm}{\linty (= \ltyctxp{\initty})}}
	& \tomachdottwo &
	\infer{\tjudg{}{\lambda\var.\tm}{\arr{\mty}{\linty}}} 
	{\tjudg{}{\red\tm}{\ltyctxp{\initty^k_\uppt}}}
	 \\[8pt]\hhline{======}&&&\\
	\infer{\tjudg{}{\tm\tmtwo}{\linty(= \ltyctxp{\initty}) 
		}}{\tjudg{}{\blue\tm}{\arr{\mty}{\ltyctxp{\initty^k_\downpt}}} & 
		\mset{\vdash}}
	&\tomachdotthree&
		\infer{\tjudg{}{\blue{\tm\tmtwo}}{\ltyctxp{\initty^k_\downpt}}} 
		{\tjudg{}{\tm}{\arr{\mty}{\linty}} & \mset{\vdash}}
		&
		\infer{\tjudg{}{\lambda\var.\tm}{\arr{\mty}{\linty}}} 
		{\tjudg{}{\blue\tm}{\ltyctxp{\initty^k_\downpt} (=\linty)}}
		 & \tomachdotfour &
		\infer{\tjudg{}{\blue{\lambda\var.\tm}}{\arr{\mty} 
				{\ltyctxp{\initty^k_\downpt}}}} 
		{\tjudg{}{\tm}{\linty}}
		\\[8pt]\hhline{======}\\
		\end{array}$}
		\\
	{\footnotesize$\begin{array}{ccccccc }
	\infer*{\infer{\tjudg{}{\la\var\ctxp{\var}} 
			{\arr{\mset{\myldots\linty_i\myldots}}\lintytwo}}{}}
	{\infer[i]{\tjudg{}{\red\var}{\ltyctxp{\initty^k_\uppt}_i (= \linty_i)}}{}}   
	&\tomachvar&
	 \infer*{\infer{\tjudg{}{\blue{\la\var\ctxp{\var}}} 
			{\arr{\mset{\myldots\ltyctxp{\initty^k_\downpt}_i\myldots}}\lintytwo}}{}}
	{\infer[i]{\tjudg{}{\var}{\linty_i}}{}}
	\\[8pt]\hhline{===}\\
	
	\infer*{\infer{\tjudg{}{\red{\la\var\ctxp{\var}}} 
			{\arr{\mset{\myldots\ltyctxp{\initty^k_\uppt}_i\myldots}}\lintytwo}}{}}
	{\infer[i]{\tjudg{}{\var}{\linty_i (=\ltyctxp{\initty}_i)}}{}}
	 & \tomachbttwo &
	 \infer*{\infer{\tjudg{}{\la\var\ctxp{\var}} 
			{\arr{\mset{\myldots\linty_i\myldots}}\lintytwo}}{}}
	{\infer[i]{\tjudg{}{\blue\var}{\ltyctxp{\initty^{k+1}_\downpt}_i}}{}} 
	\\[8pt]\hhline{===}\\
		\infer{\tjudg{}{\tm\tmtwo}{\lintytwo}} 
		{\tjudg{}{\blue\tm}{\arr{\mset{\myldots 
						\ltyctxp{\initty^k_\downpt}_i\myldots}}{\lintytwo}}
			& \tjudgi{}{\tmtwo}{\linty_i (=\ltyctxp{\initty}_i)}}
		& \tomacharg &
		\infer{\tjudg{}{\tm\tmtwo}{\lintytwo}} 
		{\tjudg{}{\tm}{\arr{\mset{\myldots 
						\linty_i\myldots}}{\lintytwo}}
			& \tjudgi{}{\red\tmtwo}{\ltyctxp{\initty^k_\uppt}_i}}
		\\[8pt]\hhline{===}\\
		\infer{\tjudg{}{\tm\tmtwo}{\lintytwo}} 
		{\tjudg{}{\tm}{\arr{\mset{\myldots 
						\linty_i\myldots}}{\lintytwo}}
			& \tjudgi{}{\blue\tmtwo}{\ltyctxp{\initty^{k+1}_\downpt}_i (=\linty_i)}}
		 & \tomachbtone &
		\infer{\tjudg{}{\tm\tmtwo}{\lintytwo}} 
		{\tjudg{}{\red\tm}{\arr{\mset{\myldots 
						\ltyctxp{\initty^k_\uppt}_i\myldots}}{\lintytwo}}
			& \tjudgi{}{\tmtwo}{\linty_i}}
		\\[8pt]\hhline{===}\\
		\infer{\tjudg{}{\tm\tmtwo}{\lintytwo}} 
			{\tjudg{}{\tm}{\arr{\mset{\myldots 
				\linty_i\myldots}}{\lintytwo}}
			& q \defeq\, \tjudgi{}{\blue\tmtwo}{\ltyctxp{\initty^0_\downpt}_i (=\linty_i)}}
			\qquad & \tomachjump & \qquad
			\jump{q} =\, \infer{\tjudg{}{\blue\var}{\ltyctxp{\initty^0_\downpt}_i (=\linty_i)}}{}\\[5pt]
			\multicolumn{3}{l}{} \\[3pt]
				\multicolumn{3}{l}{ }
		\end{array}$}\\[-12pt]
	where  $\jump{\tjudg{}{\blue\tmtwo}{\ltyctxp{\initty^0_\downpt} }}$ is the first encountered state $\tjudg{}{\blue\var}{\ltyctxp{\initty^0_\downpt} }$  such that\\$\tjudg{}{\blue\tmtwo}{\ltyctxp{\initty^\infty_\downpt} } \toksjam^+ \tjudg{}{\blue\var}{\ltyctxp{\initty^\infty_\downpt} }$.
		\end{tabular}

}
	\caption{The transitions of the Sequence Parametric \JAM (\KSJAM).}
	\label{fig:sjam-strat}  
\end{figure}
Then, it moves from judgment to judgment, following occurrences of $\initty$ around $\tyd$. The 
transitions are in \reffig{sjam-strat}, their union noted $\toksjam$. They have the same 
labels as \KJAM transitions, because they correspond to each other, as we shall show. As in the case of the \KJAM, the \KSJAM is deterministic.

The \KSJAM is a parametric version of the \SIAM of~\cite{POPL2021}. Since most of its transitions coincide with those of the original machine, we refer to that work for a detailed explanation. We discuss the behavior of the parameter $k$. As in the \KJAM, the parameter $k$ is decremented after a transition $\tomachbtone$ and incremented after a transition $\tomachbttwo$. When $k$ reaches $0$, the backtracking mechanism is disabled and the $\tomachjump$ transition is enabled. This means that the transitions $\tomachbtone$ and $\tomachbttwo$, are substituted by the  non-local ``macro'' transition $\tomachjump$ that simulates the jump of the \KJAM. In fact, on type derivations we do not have logged positions, and thus we cannot define the jump in a direct way.
Please note that \emph{a priori} $\jump{\cdot}$ as defined in \reffig{sjam-strat} is a partial function. Indeed, it is not obvious that there always exists a run such as the one described in the definition of the $\jump{\cdot}$ function. We will prove in the following section that indeed $\jump{\cdot}$ is total if restricted to reachable states.


\begin{example}
\label{ex:skjam}
We present here the execution on the \KSJAM on the same term analyzed in \refsect{k-machine}, with parameter $\kparam=1$. We have reported its type derivation, with the occurrences of $\star$ on the right of $\vdash$ annotated with increasing integers, a direction and a backtracking parameter. The occurrence of $\star$ marked with $1$ represents the first state, and so on.
\[
\infer{\tjudg{}{(\la\var\var\var)(\la\vartwo\vartwo)}{\initty^1_{\uppt{\red{1}}}}}{
	\infer{\tjudg{}{\la\var{\var\var}}{\arr{\mset{\arr{\mset{\initty^0_{\uppt{\red
									{9}}}}}{\initty^1_{\downpt{\blue
								{5}}}},\initty^1_{\downpt{\blue 
							{12}}}}}{\initty^1_{\uppt{\red
						{2}}}}}}{
		\infer{\tjudg{\var:\mset{\arr{\mset{\initty}}{\initty},\initty}}{\var\var}{\initty^1_{\uppt{\red
						{3}}}}}{
			\infer{\tjudg{\var:\mset{\arr{\mset{\initty}}{\initty}}}{\var}{\arr{\mset{\initty^1_{\downpt{\blue
									{10}}}}}{\initty_{\uppt{\red
								{4}}}}}}{}
			& \infer{\tjudg{\var:\mset\initty}{\var}{\initty^1_{\uppt{\red 
							{11}}}}}{}}} &
	\infer{\tjudg{}{\la\vartwo\vartwo}{\arr{\mset{\initty^1_{\downpt{\blue 
							{8}}}}}{\initty^1_{\uppt{\red
						{6}}}}}}{
		\infer{\tjudg{\vartwo:\initty}{\vartwo}{\initty^1_{\uppt{\red 
						{7}}}}}{}} &
	\infer{\tjudg{}{\la\vartwo\vartwo}{\initty^1_{\uppt{\red 
					{13}}}}}{}}
\]
One can immediately notice that every occurrence of $\star$ is visited exactly once. Also, the
sequence of the visited subterms and transitions is the same as the one obtained in the example of \refsect{k-machine}.
As we did for the example on the \KJAM, we now provide an execution with initial parameter $\kparam=0$.
\[
\infer{\tjudg{}{(\la\var\var\var)(\la\vartwo\vartwo)}{\initty^0_{\uppt{\red{1}}}}}{
	\infer{\tjudg{}{\la\var{\var\var}}{\arr{\mset{\arr{\mset{\initty}}{\initty^0_{\downpt{\blue
								{5}}}},\initty^0_{\downpt{\blue 
							{11}}}}}{\initty^0_{\uppt{\red
						{2}}}}}}{
		\infer{\tjudg{\var:\mset{\arr{\mset{\initty}}{\initty},\initty}}{\var\var}{\initty^0_{\uppt{\red
						{3}}}}}{
			\infer{\tjudg{\var:\mset{\arr{\mset{\initty}}{\initty}}}{\var}{\arr{\mset{\initty^0_{\downpt{\blue
									{9}}}}}{\initty^0_{\uppt{\red
								{4}}}}}}{}
			& \infer{\tjudg{\var:\mset\initty}{\var}{\initty^0_{\uppt{\red 
							{10}}}}}{}}} &
	\infer{\tjudg{}{\la\vartwo\vartwo}{\arr{\mset{\initty^0_{\downpt{\blue 
							{8}}}}}{\initty^0_{\uppt{\red
						{6}}}}}}{
		\infer{\tjudg{\vartwo:\initty}{\vartwo}{\initty^0_{\uppt{\red 
						{7}}}}}{}} &
	\infer{\tjudg{}{\la\vartwo\vartwo}{\initty^0_{\uppt{\red 
					{12}}}}}{}}
\]
Notice that the two runs are identical up to the eighth state. In the latter example, the backtracking depth of this state is $0$, so the machine jumps directly to the axiom corresponding to $x$. In the former on the other hand, the machine initiates a backtracking run, passing once again through the subterm $\lambda x.xx$ before reaching the axiom.
\end{example}

The \KJAM and the \KSJAM are strongly bisimilar. The core idea underlying this bisimulation is that each state of the \KSJAM encodes the log of the \KJAM in the judgment occurrence and the tape in the type context. Since the proof of this correspondence essentially follows the same argument as in~\cite{POPL2021,LICS2021}, we defer the technical details to the Appendix and present here only its main consequence.

\begin{proposition}[\KJAM-\KSJAM Bisimulation]
    \label{prop:bisim-sjam-ljam-art}
	The $\mathrm\KJAM$ and the $\mathrm\KSJAM$ are strongly bisimilar.
\end{proposition}

In particular, inside the technical development one proves that, whenever a $\tomachjump$ transition is performed from a reachable \KSJAM state, there always exists a run respecting the definition of the $\jump{\cdot}$ function given above.

\begin{restatable}[Well-Definedness of $\mathsf{Jump}$]{proposition}{cjumpwelldef}
	\label{prop:jump-well-def}
	Let $\state$ be a reachable $\mathrm\KSJAM$ state. Then $\jump{\state}$ is well-defined.
\end{restatable}
\section{Measuring the \KJAM Runtime via Sequence Types}
\label{sect:measures}
In this section, we introduce the weight assignment system $\WeightTimekJAM{\pi}{k}$ and show that the weight of a typed $\lambda$-term $\tm$ corresponds exactly to the length of the complete \KJAM run on $\tm$.
The main idea we rely on is that some occurrences of $\star$ in types appearing on the right-hand side of judgments within a derivation $\pi \pof\vdash t : \star$ correspond to the states of the \KJAM run starting from $t$. Therefore, to determine the number of states reached by this run, we have to count the number of times such occurrences of $\star$ appear in $\pi$.
As we have seen, the $\tomachjump$ transition enables the machine to skip certain parts of the computation through a mechanism that depends on the initial parameter $\kparam$. In Example~\ref{ex:skjam}, we see that when the initial parameter is $k=0$, the occurrences of $\star$ that are visited are precisely those located within at most one sequence. In contrast, when $k=1$, the machine also reaches an occurrence of $\star$ nested within two sequences. We must therefore study how the presence of this parameter affects the occurrences of $\star$ in $\pi$ that are reachable by the \KSJAM from an initial state $\state^{k<\infty}_\tm$.

We first define formally the notion of depth of a type context:
	\[
	\depth{\ctxhole} \defeq 0 \quad \depth{\arr\mty\ltyctx} \defeq \depth{\ltyctx}
	\quad \depth{\arr{\mset{\myldots\ltyctx\myldots}}\linty}\defeq 1+ \depth{\ltyctx}
	\]
Then, we prove that all the \KSJAM states which are reachable from $\state^{k<\infty}_\tm$ have a type context of depth at most $2\kparam + 1$. As explained in the previous section, the bisimulation relation connects the type context of a \KSJAM state with the tape of the corresponding \KJAM state. In particular, given a \KSJAM\ state $\tjudg{}{\tmtwo}{\ltyctxp{\initty^k_\pol}}$, the number of logged positions on the tape of the corresponding \KJAM\ state is exactly $\depth{\ltyctx}$.
In Example~\ref{ex:skjam}, we observe a similar phenomenon in the \KSJAM. When $k=0$, the states reached by the \KJAM have at most one logged position in the tape. In contrast, when $k=1$, a state with two logged positions on the tape is reached. The following result can be seen as the type-theoretic counterpart of~\refprop{tape-logs-jam-strat}.
\begin{restatable}[\KSJAM Phase Invariant]{proposition}{propdepthcontextksjam}
	\label{prop:depth-context-SJAM-strat}
	Let $\state\defeq\,\tjudg{}{\tm}{\ltyctxp{\initty_\pol^p} }$ be a $\mathrm\KSJAM$ state reachable by the initial state $\state^{k}_u$.
	 If $\pol=\uppt$, then $\depth{\ltyctx}=2(\kparam-p)$, otherwise if $\pol=\downpt$, then $\depth{\ltyctx}=2(\kparam-p)+1$. In particular, $\depth{\ltyctx}\leq 2k+1$.
\end{restatable}

We define a norm on linear and sequence types, which counts the number of occurrences of the ground type $\initty$ inside at most $k\geq 1$ nested sequences:
\[\begin{array}{c@{\hspace{1cm}}c}
	\begin{array}{rcl}\occstar{\initty}_k &\defeq& 1\\[3pt]
	\occstar{\arr{\mty}{\linty}}_k &\defeq& \occstar{\mty}_k+\occstar{\linty}_k 
\end{array}
&
	\occstar{\mset{\linty_1,\ldots,\linty_n}}_k \defeq 
    	\begin{cases}
      		n & \text{if $k=1$}\\
      		\sum_{1\leq i \leq n}\occstar{\linty_i}_{k-1} & \text{if $k>1$}
    	\end{cases}	
\end{array}
  \]
This is extended to type derivations with the weight system $\WeightTimekJAM{\cdot}{k}$ of~\reffig{mtypesystem}.
\begin{proposition}[Weights Bound \KJAM Time]
	\label{prop:low-bound}
	Let $\tyd\pof\tjudgw{}{}{\tm}{\initty}$ be a type derivation and $n$ the length of the complete $\mathrm\KJAM$ run from the initial state $\state_\tm^\kparam$. Then $n\leq \WeightTimekJAM{\pi}{2\kparam+1}$.
\end{proposition}
\begin{proof}
	By the definition of the weight assignment system, $\WeightTimekJAM{\pi}{2\kparam+1}$ is the number of occurrences of $\initty$ in $\tyd$ having depth less or equal than $2\kparam + 1$. Since by the proposition above all  reachable \KSJAM states having depth less or equal than $2\kparam + 1$, the result follows by the bisimulation between the \KJAM and the \KSJAM (Prop.~\ref{prop:bisim-sjam-ljam-art}).
\end{proof}
Of course, a priori it could be the case that there are occurrences of $\initty$ at depth less or equal than $2\kparam+1$ and which are \emph{not} reachable \KSJAM states. We should rule out this case in order to turn the upper bound of the proposition above into an exact measure. We prove this fact in several steps. First, we need to recall a result about the \SIAM, which corresponds in our setting to the \KSJAM with initial parameter $\infty$.

\begin{proposition}[All $\initty$ Are Reachable, \cite{POPL2021}]
Let $\tyd \pof \tjudg{}{\tm}{\initty}$ be a type derivation. Then every occurrence of $\initty$ corresponds to a state reachable by the $\mathrm\KJAM$ from the initial state $\state^\infty_\tm$.
\end{proposition}

If we fix $k\neq \infty$, the occurrences of $\star$ reachable by the \KSJAM from the initial state $\state^k_\tm$ are clearly a subset of those reachable from the initial state $\state^\infty_\tm$. In particular, by the very definition of the \KSJAM, it is clear that those occurrences of $\star$ that are reachable from $\state^\infty_\tm$ but not from $\state^k_\tm$ correspond precisely to the states that are \emph{jumped}. Moreover, whenever a $\tomachjump$ transition is performed starting from $\state^k_\tm$, there exists a corresponding run starting from $\state^\infty_\tm$ that reaches those jumped states. Having observed this, it is enough to prove that all such jumped states have depth strictly greater than $2\kparam+1$.
We prove this result in the appendix, by strengthening an invariant needed in the bisimulation proof.

\begin{restatable}[\KSJAM Jump Runs Depth]{proposition}{simuljmpsiamsjam}
\label{prop:simul-jmp-siam-ksjam}
	Let $\statetwo$ be a state reachable by the $\mathrm\KSJAM$ with initial parameter $\kparam$ and $\statetwo\eqdef (\tyd, \ruleoc', \ltyctx, \downpt,0) \tomachjump (\tyd, \ruleoc, \ltyctx, \downpt,0) \defeq \state$. Then there exists a run $\rho: (\tyd, \ruleoc', \ltyctx, \downpt,\infty) \tomachbtone s\toksjam^*s'\tomachbttwo (\tyd, \ruleoc, \ltyctx, \downpt,\infty)$ such that for each state $s''=(\tyd, \ruleoc'', \ltyctxtwo, \pol,\infty)$ in $\sigma: s\toksjam^*s'$ we have $\depth \ltyctxtwo>2\kparam+1$.
\end{restatable}

\begin{corollary}[\KJAM Time Bounds Weights]
	\label{coro:up-bound}
	Let $\tyd\pof\tjudgw{}{}{\tm}{\initty}$ be a type derivation and $n$ the length of the complete $\mathrm\KJAM$ run from the initial state $\state_\tm^k$. Then $\WeightTimekJAM{\pi}{2\kparam+1}\leq n$.
\end{corollary}
\begin{proof}
	By the proposition above, we have that all the \KSJAM unreachable occurrences of $\initty$ in $\tyd$ have depth $> 2\kparam+1$. This means that all the occurrences of $\initty$ which have depth $\leq 2\kparam+1$ are reachable. The result then follows by the bisimulation between the \KJAM and the \KSJAM (Prop.~\ref{prop:bisim-sjam-ljam-art}).
\end{proof}

The weight system is then capable of precisely measuring the runtime of the \KJAM.

\begin{restatable}[\KJAM Time via Weighted Derivations]{theorem}{mainthmstrat2}
	\label{thm:main-thm-strat}
	Let $\tm$ be a closed term. The following are equivalent:
	\begin{enumerate}
		\item There exists a complete $\mathrm\KJAM$ run $\rho$ from the initial state $\state^k_\tm$ such that $\size{\rho}=n$.
		\item There exists a derivation $\pi$ such that $\tyd\pof\tjudgw{}{}{\tm}{\initty}$ and $\WeightTimekJAM{\pi}{2\kparam+1}= n$.
	\end{enumerate}
\end{restatable}
\begin{proof}
	$(2)\Rightarrow(1)$. We argue by using the bisimilarity between the \KJAM and the \KSJAM together with \refprop{low-bound} and \refcoro{up-bound}.
	$(1)\Rightarrow(2)$. Since the \KJAM is correct for Closed CbN, as stated in \refprop{implements-cbn}, the existence of a finite \KJAM run implies that $t$ is has a normal form. This, by \refprop{bound-to}, implies the existence of a derivation $\pi$ such that $\tyd\pof\tjudgw{}{}{\tm}{\initty}$. We are therefore once again in the setting of both \refprop{low-bound} and \refcoro{up-bound}.
\end{proof}
\section{Complexity of the Parametric \JAM}
\label{sect:complexity}
In this last section, we exploit our type theoretic analysis of the \KJAM to study its complexity. To achieve this, we establish a bound on the number of occurrences of the base type $\star$ up to a fixed depth within each type in a given derivation. Let $t$ be a closed term and $\pi$ be a derivation such that $\pi\pof\vdash t:\star$, also let $n$ be the number of $\beta$ steps needed to normalize $t$. We show that the number of occurrences of $\star$ in $\pi$ that appear at depth at most $k$ is polynomial in $n$, for any $k\in\nat$. This result, in conjunction with \refthm{main-thm-strat}, yields a bound on the length of the execution of the machine, proving that the complexity of the Parametric \JAM is polynomial with respect to $n$ for any given initial parameter $k<\infty$.

\subparagraph{Bounding Arrows and Sequences.}
We begin by defining a size on types capturing the maximum number of arrows ``at the same depth''. Intuitively, viewing a type $A$ as a tree in which internal nodes correspond to $\rightarrow$ and in which leaves correspond to occurrences of $\star$, we define the size $\sizeto A$ as the maximal length of a path from the root to any leaf.
\[\sizeto \star:=0 \qquad \sizeto {S\to A}:=\max(\sizeto S,1+\sizeto A) \qquad \sizeto{[A_1,\dots,A_n]}:=\max_{i\leq n}(\sizeto {A_i})\]
	We define $\maxto \pi$ as the maximum $\sizeto A$ for any type $A$ occurring in $\pi$ (on the right-hand side of $\vdash$).
Since by $\beta$ expanding a term at most one occurrence of $\to$ can be created, we have the following bound.
%
\begin{restatable}[Arrow Size Bound]{proposition}{propboundto}
	\label{prop:bound-to}
	Let $t$ be a closed $\lambda$-term such that $t \towh^n \lambda x.u$. Then, there exists $\pi\pof \vdash t:\star$ such that $\maxto{\pi}\leq n$.
\end{restatable}
%
%
Then, we need to define a new measure on types, that computes the maximum size of sequences appearing in a type:
%
\[\sizemult \star:=0 \qquad \sizemult{S\to A}:=\max(\sizemult S,\sizemult A) \qquad \sizemult{[A_1,\dots,A_n]}:=\max(\max_{i\leq n}\sizemult{A_i},n)\]
	We define $\maxmult\pi$ as the maximum size $\sizemult{A}$ for any type $A$ occurring in $\pi$ (on the right-hand side of $\vdash$).
We observe that that the maximum cardinality of sequences in a type derivation $\tyd$ ending in $\initty$ cannot be larger than the size of $\tyd$ itself, noted $\size{\tyd}$.
%
%
\begin{restatable}[Sequence Size Bound]{proposition}{propboundm}	
	\label{prop:bound-m}
	Let $t$ be a closed term such that $t \towh^* \lambda x.u$ Then, there exists $\pi\pof \vdash t:\star$ such that $\maxmult{\pi}\leq \size \pi$.
\end{restatable}
%
Given a term $t$ normalizing in $n$ steps, we are able to transfer the bound from being parametric on the size of its type derivation, to $n$. The overhead is indeed at most quadratic.

\begin{proposition}[Derivation Size Bound~\cite{deCarvalho18,DBLP:conf/ppdp/AccattoliB17}]
	\label{prop:bound-size-beta}
	Let $t$ be a closed term such that $t \towh^n \lambda x.u$. Then there exists $\pi\pof \vdash t:\star$ such that $\size{\pi} = \bigo {n^2}$.
\end{proposition}

\subparagraph{Final Bound.}

We can now give a bound on the number of occurrences of $\star$ that appear inside at most $k$ nested sequences in a type $A$, \emph{\ie} the measure $\occstar{A}_k$. Each type $A$ has the shape $A=[A_1^1,\dots,A_{m_1}^1]\to \dots \to [A_1^n,\dots,A_{m_n}^n]\to \star$, for $n,m \geq 0$. By definition, we have that $n\leq \sizeto{A}$ and $m_i\leq\sizemult{A}$ for each $i\leq n$.\footnote{Notice that $n$ could in fact be bound by the maximal number of $\to$ occurring at the same sequence depth, whereas our size $\sizeto{\cdot}$ is larger. We use this definition because it leads to the same final bound and it is simpler to work with. } Hence, the number of types $A_j^i$ is at most $\sizeto{A}\cdot \sizemult{A}$. Notice that each $A_j^i$ has exactly one $\star$ at depth $0$, which appears at depth $1$ in $A$. Hence, we obtain the bound $\occstar{A}_1\leq \sizeto{A}\cdot \sizemult{A} +1$. 
By iterating this bound on each type $A_j^i$, we get the following lemma.

\begin{restatable}[Type Size Bound]{lemma}{lemmaboundtypesizek}
\label{l:bound-type-size-k}
	For each type $\linty$ and each $k\geq 0$ we have that\\ $\occstar{\linty}_{k+1}\leq (\sizeto{A}\cdot \sizemult{A})^{k+1} + \occstar{\linty}_{k}$.
\end{restatable}

Then, we obtain a bound as a function of the number of normalization steps.

\begin{restatable}[Concrete Type Bound]{proposition}{coroboundoccstartype}
\label{p:bound-occstar-type}
	Let $t$ be a closed term such that $t\towh^n\lambda x.u$ and $\pi\pof \vdash t:\star$. Then, for each type $\linty$ in $\pi$ and $k<\infty$ we have
		$\occstar{A}_{k} = \bigo{n^{3k}}$. 
\end{restatable}
\begin{proof}
	We proceed by induction on $k$. If $k=0$ then $\occstar{A}_0=1=n^0$. If $k>0$:
	\begin{align*}
		\occstar{A}_{k}&\stackrel{(\text{\reflemmaeq{bound-type-size-k}})}{\leq} (\sizeto{A}\cdot \sizemult{A})^{k} + \occstar{A}_{k-1} \stackrel{(\text{\refprop{bound-to} + \emph{\ih}})}{\leq} (n\cdot \sizemult{A})^{k} + \bigo{n^{3(k-1)}}\\
        &\stackrel{(\text{\refprop{bound-m} + \refprop{bound-size-beta}})}{\leq} (n\cdot \bigo{n^2})^{k} + \bigo{n^{3(k-1)}}
		= \bigo{n^{3k}}.\qedhere
	\end{align*}
\end{proof}
Then, we count the number of occurrences of $\star$ appearing at most at depth $k<\infty$ in the entire derivation $\tyd$. This amounts to multiplying the bound established in the previous proposition by the size of $\tyd$ (quadratic in the number of $\beta$-steps, see Prop.~\ref{prop:bound-size-beta}).

\begin{restatable}[Bounding Weighted Derivations]{theorem}{thmboundpik}	
	\label{thm:bound-pik}
	Let $t$ be a closed term such that $t\towh^n\lambda x.u$ and $\pi\pof \vdash t:\star$, then  
	$
    \WeightTimekJAM{\pi}{k}= \bigo{n^{3k+2}}
	$ for each $k<\infty$.
\end{restatable}
Finally, we obtain a bound on the length of complete \KJAM runs.
\begin{corollary}[The \KJAM Has Polynomial Overhead]
	\label{coro:main-bound-run}
	Let $t$ be a closed term such that $t\towh^n\lambda x.u$. Then the complete $\mathrm{\KJAM}$ run from $\state^k_\tm$ has length $\bigo{n^{6k+5}}$, for each $k<\infty$.
\end{corollary}
\begin{proof}
	Let $m$ be the length of the complete run. By \refthm{main-thm-strat}, we have $\pi\pof\vdash t:\star$ such that
	$
		m\stackrel{(\text{\refthm{main-thm-strat}})}{=}\WeightTimekJAM{\pi}{2k+1}\stackrel{(\text{\refthm{bound-pik}})}{=}\bigo{n^{3(2k+1)+2}}=\bigo{n^{6k+5}}$. \qedhere
\end{proof}
Now, the natural question is if this bound is tight. We believe it is not, as in the case of the \JAM, \ie when $k=0$, we already know from~\cite{POPL2021} that another bound is $\bigo{n^4\cdot\size{t}}$, likely not tight as well, instead of $\bigo{n^5}$. We leave for future work understanding the relationship between our complexity analysis and the one carried out in~\cite{POPL2021}, which seem to be orthogonal. Intuitively, here we have analyzed type derivations ``horizontally'', judgment by judgment, while in~\cite{POPL2021} they have been analyzed ``vertically'', hence the dependence on the size of the term, which bounds derivation height.

\section{Conclusion}
In this paper, we introduced a new parametric machine, the \KJAM, which behaves like the \IAM up to a fixed backtracking nesting depth, and then switches to behaving like the \JAM. We showed that the runtime of this machine can be characterized using intersection type derivations, extending the results about the \KAM and the \IAM. Furthermore, thanks to the aforementioned type theoretic framework, we were able to analyze its complexity. This allowed us to prove that the number of steps the \KJAM needs to evaluate a term $\tm$ is \emph{polynomial} in the number of $\beta$-steps needed to normalize it, thus making the \KJAM a \emph{reasonable} cost model~\cite{DBLP:journals/entcs/Accattoli18}.

Although, from a technical viewpoint, the analysis of the PaJAM just described is certainly of interest, the most significant conceptual contribution lies in the exploration of the space between the JAM and the IAM, the former being obtained as an optimization of the latter and thereby achieving, on certain terms, an exponential speedup. On the one hand, our work shows that the phase transition occurs ‘on the IAM side,’ only when jumps are never performed. On the other hand, it demonstrates that, from a quantitative perspective, the JAM is more similar to the KAM than to the machine of which it is an optimization.

In future work, we would like to investigate the relationship between the \KJAM and game semantics. Indeed, it would be interesting to analyze our results on non-idempotent intersection types from the perspective of game semantics, \eg following the lines of~\cite{TsukadaAO17,DBLP:conf/lics/ClairambaultF24}, and understand how the \KJAM relates with game-theoretical models, enriching the already well-known connections between the \JAM and HO games (via the PAM-JAM isomorphism) and the \IAM and AJM games.

Moreover, we are interested in the \emph{space} performances of the \KJAM, in particular depending on the choice of the parameter. This analysis could be carried out either directly on the machine, as in~\cite{LICS2022,LMCS2024}, or via intersection types, as in~\cite{LICS2021,ICALP2022}. The problem is non-trivial, as one should in this case take into account how data structures are represented at low-level, and how garbage collection is implemented.

\bibliography{refs}

\newpage

\appendix

\section{Proof of the J-exhaustible invariant}
\label{sect:apx4}

This section contains proofs which are rephrasing of proofs appearing in~\cite{POPL2021}. We added them for the sake of completeness. Please note, however, the strengthening of Definition~\ref{def:j-exhaust-strength}.
\subparagraph{The Exhaustible Invariant.} This technique was first introduced in~\cite{PPDP2020}, where it was used to prove the correctness of the \IAM, and later adapted to intersection type-based machines in~\cite{POPL2021}, to prove the correspondence with automata-theoretic ones. Roughly, the exhaustible invariant says that reachable states are \emph{consistent}, in the sense that all the data structures are coherent. Here, in the type theoretic formulation of the \KJAM, we use it to prove that the function $\jump{\cdot}$ is always defined on reachable states, and to allow us to define the bisimulation between the \KJAM and the \KSJAM. The definition of the invariant itself is highly technical, and requires some auxiliary concepts.

\subparagraph{Preliminaries for the Invariant.} First, we need to generalize the \KSJAM to arbitrary positions. This means that the token does not necessarily move on an occurrence of $\initty$ any more, but it can travel on any linear type $\linty$. A \KSJAM state 
is a quintuple $(\tyd, \ruleoc, \ltyctx, \pol,k)$ where $\ruleoc$ is an 
occurrence 
of a judgment 
$\tjudg{\tye}{\tmtwo}{\linty}$ in $\tyd$, $\pol$ is a direction, $k$ is a parameter, and $\ltyctx$ 
is a type context isolating an occurrence of 
$\initty$ in $\linty$. The generalization simply is to consider type contexts 
$\ltyctx$ such that $\ltyctxp{\lintytwo} = \linty$ 
for some $\lintytwo$, that is, not necessarily isolating $\initty$. A pair 
$(\lintytwo, \ltyctx)$ such that $\ltyctxp{\lintytwo} = 
\linty$ is called a position in $\linty$. 

Note that the \KSJAM can be naturally adapted to this more general notion of state, that follows an arbitrary formula 
$\lintytwo$, not necessarily $\initty$. 
It amounts to simply replacing $\initty$ 
with $\lintytwo$. We give the transitions of the Generalized \KSJAM in \reffig{ksjam-gen}.
To easily manage \KSJAM states, we also use a concise notation, writing 
$\tjudg{}{\tm}{\linty^k_d,\ltyctx}$ for a state 
$\state= (\tyd, \ruleoc, (\linty,\ltyctx), \pol,k)$ where $\ruleoc$ is 
$\tjudg{\tye}{\tm}{\ltyctxp\linty}$ for some $\tye$, potentially 
specifying the direction via colors and under/over-lining.
As we did in the \KSJAM, we define $\jump{\tjudg{}{\blue\tmtwo}{\ltyctxp{\lintyb^0_\downpt} }}$ as the first encountered state $\tjudg{}{\blue\var}{\ltyctxp{\lintyb^0_\downpt} }$ such that  $\tjudg{}{\blue\tmtwo}{\ltyctxp{\lintyb^{\infty}_\downpt} } \toksjam^+ \tjudg{}{\blue\var}{\ltyctxp{\lintyb^{\infty}_\downpt} }$.

\subparagraph{\SIAM Tests.} Given a \SJAM state $\state= (\tyd, \ruleoc, 
(\linty,\ltyctx), \pol,k)$, the underlying idea is that 
the judgment occurrence $\ruleoc$ encodes the log of the \KJAM, while the type 
context $\ltyctx$ encodes the 
tape. It is then natural to define two kinds of test, one for judgments and 
one for type contexts.

The intuition is that a test focuses on (the occurrence of) an element 
$\lintytwo$ 
of a sequence $\mty$ related to 
$\state$, and that these sequence elements play the role of logged positions in the \KJAM. These sequence elements 
are of two kinds:
\begin{enumerate}
	\item \emph{Elements containing $\ruleoc$}: those in which the focused 
	judgment $\ruleoc$ itself is contained, 
	corresponding to the logged positions in the log of the \KJAM. Note that the positions on the log are those for which 
	the \KJAM has previously found the corresponding arguments. In the \KSJAM these arguments are exactly those in which the 
	focused judgment is contained.
	
	\item \emph{Elements appearing in $\ltyctx$}: those in 
	the right-hand type of $\state$ in which the focused type $\linty$ is 
	contained, 
	corresponding to the logged positions on 
	the tape of the \KJAM. They correspond to \KJAM queries for which the argument has not yet been found.
\end{enumerate}
Each one of these elements is then identified by a judgment occurrence 
$\ruleoc'$ and a position $(\lintytwo,\ltyctxtwo)$ 
in the right-hand type of $\ruleoc'$. 

\begin{definition}[Focus]
	A \emph{focus} $\focus$ in a derivation $\tyd$ is a pair $\focus = (\ruleoc, 
	(\linty,\ltyctx))$ of a judgment occurrence 
	$\ruleoc$ and of a type position $(\linty,\ltyctx)$ in the right-hand type 
	$\ltyctxp\linty$ of $\ruleoc$.
\end{definition}

The intuition is that exhausting a test $\state_{\ruleoc, (\linty,\ltyctx)}$ in 
$\tyd$ shall amount to retrieve the 
axiom of 
$\tyd$ of type $\linty$ that would be substituted by that sequence element of 
type $\linty$ by reducing $\tyd$ via 
cut-elimination---the 
definition of exhaustible tests is given below, after the definition of tests.

\begin{definition}[Judgment Tests]
	Let $\state=(\tyd, \ruleoc, (\linty,\ltyctx), \pol,k)$ be a \KSJAM state. Let 
	$r_i$ be $i$-th $\tyapp$ rule found traversing $\tyd$ 
	by descending from the focused judgment $\ruleoc$ towards the final judgment 
	of $\tyd$. Let $\ruleoc_i$ be the 
	judgment of the sequence $\mty_i$ in the right premise of $r_i$ traversed in 
	such a descent (careful: $\ruleoc_i$ is 
	the $j$-th judgment of $\mty_i$ for some $j$, that is, the index $i$ denotes 
	the connection with rule $r_i$, not the 
	position in $\mty_i$). 
	Let $\ruleoc_i$ be $\tjudg{\tye}{\tm}{\lintytwo}$. Then 
	$\state_{\focus}^i = (\tyd,\ruleoc_i,(\lintytwo,\ctxhole), 
	\downpt,\infty)$ is the $i$-th judgment test of $\state$, having as focus $\focus 
	\defeq (\ruleoc_i,(\lintytwo,\ctxhole))$. 
\end{definition}
We 
often
omit the judgment from the focus, writing simply 
$\state_{(\lintytwo,\ctxhole)}$, 
and even concisely 
note $\state_{\focus}$ as 
$\tjudg{}{\blue\tm}{\lintytwo_\downpt,\ctxhole}$.
Note that judgment tests always have type context $\ctxhole$.

\subparagraph{Type (Context) Tests.} While judgment tests depend only on the 
judgment occurrence $\ruleoc$ of a state 
$\state = (\tyd, \ruleoc, (\linty,\ltyctx), \pol,k)$, type context 
tests---dually---fix $\ruleoc$ and depend only on the 
type 
context $\ltyctx$ of $\state$, that is, they all focus on sequence elements of 
the form $(\ruleoc, (\lintytwo,\ltyctxtwo))$ 
where $\ltyctxtwop\lintytwo = \ltyctxp\linty$ and $\ltyctx = 
\ltyctxtwop\ltyctxthree$ for some type context $\ltyctxthree$. Namely, 
there is one type context test (shortened to \emph{type test}) for every sequence 
in which the hole of $\ltyctx$ is contained. We need some notions about type 
contexts, in particular a notion of level 
analogous to the one for term contexts.

\subparagraph{Terminology About Type Contexts.} Define type contexts $\ltyctx_n$ of 
level $n\in\nat$ as follows:

\[\begin{array}{lclr}
	\ltyctx_0 &\defeq &\ctxhole \mid \arr\mty\ltyctx_0
	\\
	\ltyctx_{n+1} &\defeq &\arr{\mset{\myldots\ltyctx_n\myldots}}\linty \mid 
	\arr\mty\ltyctx_{n+1}	
\end{array}\]
Clearly, every type context $\ltyctx$ can be seen as a type context $\ltyctx_n$ 
for a unique $n$, and vice versa a type 
context of level $n$ is also simply a type context---the level is then sometimes omitted.
A \emph{prefix} of a context $\ltyctx$ is a context $\ltyctxtwo$ such that 
$\ltyctxtwop\ltyctxthree = \ltyctx$ for some 
$\ltyctxthree$. Given $\ltyctx$ of level $n>0$, there is a smallest prefix 
context $\ltyctx|_i$ of level $0<i\leq 
n$, and it has the form 
$\ltyctxtwo\ctxholep{\arr{\mset{\myldots\ctxhole\myldots}}\linty}$ for a type 
context $\ltyctxtwo$ 
of level $i-1$.


\begin{definition}[Type Tests]
	Let $\state=(\tyd, \ruleoc, (\linty,\ltyctx), \pol,k)$ be a \KSJAM state and 
	$n$ 
	be the level of $\ltyctx$. The sequence 
	of directed prefixes $\DiPref\ltyctx$ of $\ltyctx$ is the sequence of pairs  
	$(\ltyctxtwo,\poltwo)$, where $\ltyctxtwo$ is a prefix of $\ltyctx$, 
	defined as follows:
	\[\begin{array}{lclll}
		\DiPref\ltyctx & \defeq & \mset\cdot & \mbox{if }n=0
		\\
		\DiPref\ltyctx & \defeq & \mset{(\ltyctx|_1,\uppt), 
			\ldots,(\ltyctx|_n,\uppt^{n-1})} &\mbox{if }n>0	
	\end{array}\]
	The $i$-th directed prefix (from left to right) $(\ltyctxtwo,\poltwo)$ in 
	$\DiPref\ltyctx$ induces the type test 
	$\state_{\focus}^i \defeq (\tyd, \ruleoc, (\ltyctxthreep\linty,\ltyctxtwo), 
	\poltwo,\infty)$ of $\state$ and focus $\focus \defeq 
	(\ruleoc,(\ltyctxthreep\lintytwo,\ltyctxtwo))$, where $\ltyctxthree$ is the 
	unique type context such that $\ltyctx = \ltyctxtwop\ltyctxthree$.
\end{definition}

\begin{remark}
	\label{rem:type-test-i}
	Notice that, by definition, for the type test $q^i_f$ we have that $\depth{\ltyctxtwo}=i$.
\end{remark}

\begin{definition}[State Respecting a Focus]
	Let $\focus=(\ruleoc, (\linty,\ltyctx))$ be a focus. A \KSJAM state $\state$ 
	respects $\focus$ if it is an axiom 
	$\tjudg{}{\blue\var}{\ctxholep{\linty^k_\downpt}}$ for some variable $\var$ and any $k\geq 0$ 
	(the typing context of $\state$, which is omitted 
	by convention, is $\var:\mset\linty$).
\end{definition}

\begin{definition}[J-exhaustible states]
	\label{def:j-exhaust-strength}
	The set $\exstates_{J}$ of J-exhaustible states is the 
	smallest set of generalized \KSJAM states such that if $\state\defeq (\tyd, \ruleoc, (\linty,\ltyctx), \pol,k)\in\exstates_{J}$, then for each type or 
	judgment test $\state_\focus$ of $\state$ of focus $\focus$ there exists a run 	
	$\run: \state_f \toksjam^*\tomachbttwo\statetwo\defeq (\tyd, \ruleoc', (\linty',\ltyctx'), \pol',\infty)$ where $\statetwo$ respects $\focus$ and for 
	the shortest such run $(\tyd, \ruleoc', (\linty',\ltyctx'), \pol',k)\in\exstates_{J}$.
	Moreover:
	\begin{itemize}
		\item if $\state_f$ is a judgment test then $\run:  \state_f\tomachbtone s\toksjam^*s'\tomachbttwo\statetwo$ and for each state $s_0=(\tyd, \ruleoc_0, (\linty_0,\ltyctx_0), \pol_0,\infty)$ in $\sigma: s\toksjam^*s'$ we have that $\depth {\ltyctx_0}\geq 1$;
		\item if $\state_f$ is a type test then $\run:  \state_f\toksjam^*s'\tomachbttwo\statetwo$ and for each state\\ $s_0=(\tyd, \ruleoc_0, (\linty_0,\ltyctx_0), \pol_0,\infty)$ in $\sigma: \state_f\toksjam^*s'$ we have that $\depth{\ltyctx_0} \geq 1$.
	\end{itemize}
\end{definition}

\begin{remark}
	\label{rem:test-indipendent-k}
Let $\state=(\tyd,\ruleoc,(\linty,\ltyctx),\pol,k)$ and $\statetwo=(\tyd,\ruleoc,(\linty,\ltyctx),\pol,k')$, with $k,k'\geq 0$, be two states that differ only in their backtracking parameter.
Observe that, in the definition of both type and judgment tests, the backtracking parameter is always taken to be $\infty$, therefore the two states $\state$ and $\statetwo$ have exactly the same tests. Moreover, by the definition of J-exhaustible states, if $\state$ is J-exhaustible, then so is $\statetwo$ and the corresponding test run $\rho$ starting from each of their test are the same. The J-exhaustibility is then a property of the occurences of the $\star$ within the derivation, and is independent of the backtracking depth of the state. 
\end{remark}

\begin{figure}[t]
	\scalebox{.95}{\footnotesize
\begin{tabular}{c}
$\begin{array}{ccc||ccc}
	\infer{\tjudg{}{\red{\tm\tmtwo}}{\ltyctxp{\lintyb^k_\uppt} (=\linty)}} 
	{\tjudg{}{\tm}{\arr{\mty}{\linty}} & \mset{\vdash}} 
	&\tomachdotone&
	\infer{\tjudg{}{\tm\tmtwo}{\linty 
		}}{\tjudg{}{\red\tm}{\arr{\mty}{\ltyctxp{\lintyb^k_\uppt}}} & 
		\mset{\vdash}}
		&
	\infer{\tjudg{}{\red{\lambda\var.\tm}}{\arr{\mty} 
			{\ltyctxp{\lintyb^k_\uppt}}}} 
	{\tjudg{}{\tm}{\linty (= \ltyctxp{\lintyb})}}
	& \tomachdottwo &
	\infer{\tjudg{}{\lambda\var.\tm}{\arr{\mty}{\linty}}} 
	{\tjudg{}{\red\tm}{\ltyctxp{\lintyb^k_\uppt}}}
	 \\[8pt]\hhline{======}&&&\\

	\infer{\tjudg{}{\tm\tmtwo}{\linty(= \ltyctxp{\lintyb}) 
		}}{\tjudg{}{\blue\tm}{\arr{\mty}{\ltyctxp{\lintyb^k_\downpt}}} & 
		\mset{\vdash}}
		
	&\tomachdotthree&
		\infer{\tjudg{}{\blue{\tm\tmtwo}}{\ltyctxp{\lintyb^k_\downpt}}} 
		{\tjudg{}{\tm}{\arr{\mty}{\linty}} & \mset{\vdash}}
		&
		\infer{\tjudg{}{\lambda\var.\tm}{\arr{\mty}{\linty}}} 
		{\tjudg{}{\blue\tm}{\ltyctxp{\lintyb^k_\downpt} (=\linty)}}
		 & \tomachdotfour &
		\infer{\tjudg{}{\blue{\lambda\var.\tm}}{\arr{\mty} 
				{\ltyctxp{\lintyb^k_\downpt}}}} 
		{\tjudg{}{\tm}{\linty}}
		\\[8pt]\hhline{======}\\
		\end{array}$
		\\
	$\begin{array}{ccccccc }
	\infer*{\infer{\tjudg{}{\la\var\ctxp{\var}} 
			{\arr{\mset{\myldots\linty_i\myldots}}\lintytwo}}{}}
	{\infer[i]{\tjudg{}{\red\var}{\ltyctxp{\lintyb^k_\uppt}_i (= \linty_i)}}{}}   
	&\tomachvar&
	 \infer*{\infer{\tjudg{}{\blue{\la\var\ctxp{\var}}} 
			{\arr{\mset{\myldots\ltyctxp{\lintyb^k_\downpt}_i\myldots}}\lintytwo}}{}}
	{\infer[i]{\tjudg{}{\var}{\linty_i}}{}}
	\\[8pt]\hhline{===}\\
	
	\infer*{\infer{\tjudg{}{\red{\la\var\ctxp{\var}}} 
			{\arr{\mset{\myldots\ltyctxp{\lintyb^k_\uppt}_i\myldots}}\lintytwo}}{}}
	{\infer[i]{\tjudg{}{\var}{\linty_i (=\ltyctxp{\lintyb}_i)}}{}}
	 & \tomachbttwo &
	 \infer*{\infer{\tjudg{}{\la\var\ctxp{\var}} 
			{\arr{\mset{\myldots\linty_i\myldots}}\lintytwo}}{}}
	{\infer[i]{\tjudg{}{\blue\var}{\ltyctxp{\lintyb^{k+1}_\downpt}_i},}{}} 
	\\[8pt]\hhline{===}\\
		\infer{\tjudg{}{\tm\tmtwo}{\lintytwo}} 
		{\tjudg{}{\blue\tm}{\arr{\mset{\myldots 
						\ltyctxp{\lintyb^k_\downpt}_i\myldots}}{\lintytwo}}
			& \tjudgi{}{\tmtwo}{\linty_i (=\ltyctxp{\lintyb}_i)}}
		& \tomacharg &
		\infer{\tjudg{}{\tm\tmtwo}{\lintytwo}} 
		{\tjudg{}{\tm}{\arr{\mset{\myldots 
						\linty_i\myldots}}{\lintytwo}}
			& \tjudgi{}{\red\tmtwo}{\ltyctxp{\lintyb^k_\uppt}_i}}
		\\[8pt]\hhline{===}\\

		\infer{\tjudg{}{\tm\tmtwo}{\lintytwo}} 
		{\tjudg{}{\tm}{\arr{\mset{\myldots 
						\linty_i\myldots}}{\lintytwo}}
			& \tjudgi{}{\blue\tmtwo}{\ltyctxp{\lintyb^{k+1}_\downpt}_i (=\linty_i)}}
		 & \tomachbtone &
		\infer{\tjudg{}{\tm\tmtwo}{\lintytwo}} 
		{\tjudg{}{\red\tm}{\arr{\mset{\myldots 
						\ltyctxp{\lintyb^k_\uppt}_i\myldots}}{\lintytwo}}
			& \tjudgi{}{\tmtwo}{\linty_i}}

		\\[8pt]\hhline{===}\\

		\infer{\tjudg{}{\tm\tmtwo}{\lintytwo}} 
			{\tjudg{}{\tm}{\arr{\mset{\myldots 
				\linty_i\myldots}}{\lintytwo}}
			& q \defeq\, \tjudgi{}{\blue\tmtwo}{\ltyctxp{\lintyb^0_\downpt}_i (=\linty_i)}}
			\qquad & \tomachjump & \qquad
			\jump{q} =\, \infer{\tjudg{}{\blue\var}{\ltyctxp{\lintyb^0_\downpt}_i (=\linty_i)}}{}
		
		\end{array}$
		\end{tabular}
		
}
	\vspace{-8pt}
	\caption{The transitions of the Generalized \KSJAM}
	\label{fig:ksjam-gen}  
\end{figure}

\begin{lemma}[Type context lifting]
	\label{l:type-context-lifting}
	If $\rho:\;\;\tjudg{}{\tm}{\linty^\infty_d,\ltyctxb}
	\toksjam^n
	\tjudg{}{\tmtwo}{\linty^{\infty}_{d'},\ltyctxtwob}$ with $\ltyctxtwop \lintytwo=\linty$ then it exists a run $\sigma:\;\;\tjudg{}{\tm}{	\lintytwo^\infty_{d},\ltyctxbp{\ltyctxtwo}}
	\toksjam^n
	\tjudg{}{\tmtwo}{\lintytwo^\infty_{d'},\ltyctxtwobp{\ltyctxtwo}}$. Moreover, at each step the transition applied in $\rho$ and $\sigma$ are the same.
\end{lemma}
\begin{proof}
	In the following we will not take in account the parameter $k,\,k'$ in the states to make the proof simpler to read. It is easy to see that this do not change the proof.
	By induction on the length of $\rho$. If $|\rho|=0$ the result it trivially holds. Let's fix 
	\[\rho:\;\; \vdash t:\linty^\infty_{d},\ltyctxb\toksjam^n\vdash t'': \linty^\infty_{d''},\ltyctxthreeb\toksjam \vdash u:\linty^\infty_{d'}, \ltyctxtwob\] 
	By $\ih$ it exists a run 
	$\sigma':\;\; \vdash t:\lintytwo^\infty_{d},\ltyctxbp \ltyctxtwo\toksjam^n \vdash t'': \lintytwo^\infty_{d''},\ltyctxthreebp \ltyctxtwo\; =: \state''$.
	We can conclude simply by noticing that all the Generalized \KSJAM transitions are local.
	We can perform the same transition from $\vdash t'': \linty^\infty_{d''},\ltyctxthreeb$ to $\vdash u:\linty^\infty_{d'}, \ltyctxtwob$ that we had in $\rho$, in the context $\ltyctxthreebp \ltyctxtwo$.

	We give an example for the case $\tomachvar$, all the others being similar. The last transition of $\rho$ has the following shape.
	\[
	\begin{array}{clc}
		\infer*{\infer{\tjudg{}{\la\var\ctxp{\var}} 
				{\arr{\mset{\ldots\linty_i\ldots}}\lintyb}}{}}
		{\infer[i]{\tjudg{}{\red\var}{\ltyctxthreebp{\linty^\infty_{\uppt}}_i(=\linty_i)}}{}}
		& 
		\tomachvar
		& \infer*{\infer{\tjudg{}{\blue{\la\var\ctxp{\var}}}
				{\arr{\mset{\ldots\ltyctxthreebp{\linty^\infty_{\downpt}}_i\ldots}}\lintyb}}{}}
		{\infer[i]{\tjudg{}{\var}{\linty_i}}{}}	
	\end{array}
	\]
	We have $\linty=\ltyctxtwop \lintytwo$ and we can perform the same $\tomachvar$ transition on $\state''$, obtaining

	\[
	\begin{array}{clc}
		\infer*{\infer{\tjudg{}{\la\var\ctxp{\var}} 
				{\arr{\mset{\ldots\linty_i\ldots}}\lintyb}}{}}
		{\infer[i]{\tjudg{}{\red\var}{\ltyctxthreebp{\ltyctxtwop {\lintytwo^\infty_{\uppt}}}_i(=\linty_i)}}{}}
		& 
		\tomachvar
		& \infer*{\infer{\tjudg{}{\blue{\la\var\ctxp{\var}}}
				{\arr{\mset{\ldots\ltyctxthreebp{\ltyctxtwop {\lintytwo^\infty_{\downpt}}}_i\ldots}}\lintyb}}{}}
		{\infer[i]{\tjudg{}{\var}{\linty_i}}{}}=\state		
	\end{array}
	\]

	We have that the context of the new state $q$ is exactly $\ltyctxtwobp \ltyctxtwo$.
\end{proof}

\begin{lemma}[J-exhaustible Invariant]
	\label{l:K-invariant-ksjam-str}
	Let $\tm$ be a closed term, $\tyd\pof\tjudg{\tye}{\tm}{\linty}$ a sequence 
	type derivation for it, and $\run:\ \tjudg{}{\tm}{\ctxholep{{\linty}^{\kparam}_\uppt}} 
	\toksjam^n
	\state$ an initial Generalized \KSJAM run. Then $\state$ is J-exhaustible. 
\end{lemma}
\begin{proof}
	By induction on $n$. For $n=0$ there is nothing to prove because the initial 
	state $\state_0 \eqdef \,\tjudg{}{\tm}{\ctxholep{\linty^k_\uppt}}$ has no 
	judgement nor type tests. Then suppose
	$\run':\state_0\toksjam^{n-1}\statetwo$ and that the run continues with $\statetwo\toksjam\state$. By \ih, $\statetwo$ is 
	J-exhaustible.
	
	\emph{Terminology}: when a test state satisfies the clause in the definition of J-exhaustible states we say that it is \emph{positive}. 
	
	Cases of 
	$\statetwo\toksjam\state$:
		
	\begin{itemize}
		\item Case $\tomachdotone$.
		\[\begin{array}{clc}
			\statetwo=\infer{\tjudg{}{\red{\tm\tmtwo}}{\ltyctxp{\linty^k_{\uppt}} 
					(=\lintyb)}}
			{\tjudg{}{\tm}{\arr{\mty}{\lintyb}} & \mset\vdash} &
			\tomachdotone &
			\infer{\tjudg{}{\tm\tmtwo}{\lintyb 
			}}{\tjudg{}{\red\tm}{\arr{\mty}{\ltyctxp{\linty^k_{\uppt}}}} & 
				\mset\vdash}=\state
		\end{array}\]
		\begin{itemize}
			\item \emph{Judgement tests.} Note that $\state$ has the same judgement tests of $\statetwo$, which are 
			positive by the \ih
			\item \emph{Type tests.} We consider at first type tests with direction $\uppt$. Let $q_u$ be one of them. Since the type context $\ltyctx$ is the same in both $q$ and $q'$, there is a corresponding type test $q'_u$ of $q'$. This $q'_u$ is positive by \ih. We observe that $q'_u\tomachdotone q_u$, and since the machine is deterministic also $q_u$ is positive.
			We now consider type tests with direction $\downpt$, let now $q_d$ one of them. As before, there is a corresponding type test for $q'$, the we call $q'_d$ that is positive by \ih and such that $q_d\tomachdotthree q'_d$. By definition, also $q_d$ is positive. 
			In both cases, to prove the moreover part of the statement, we just have to notice that in the new transition that we defined the type context $\ltyctx$ is not changed, so neither is its depth. 
		\end{itemize}
		\item Case $\tomachdottwo$. Identical to the previous one.
		\item Case $\tomachvar$.
		\[
		\begin{array}{clc}
			\statetwo=\infer*{\infer{\tjudg{}{\la\var\ctxp{\var}} 
					{\arr{\mset{\ldots\linty_i\ldots}}\lintytwo}}{}}
			{\infer[i]{\tjudg{}{\red\var}{\ltyctxp{\linty^k_{\uppt}}_i(=\linty_i)}}{}}
			& 
			\tomachvar
			& \infer*{\infer{\tjudg{}{\blue{\la\var\ctxp{\var}}}
					{\arr{\mset{\ldots\ltyctxp{\linty^k_{\downpt}}_i\ldots}}\lintytwo}}{}}
			{\infer[i]{\tjudg{}{\var}{\linty_i}}{}}=\state		
		\end{array}
		\]
		\begin{itemize}
			\item \emph{Judgement tests.} Judgement tests of $\state$ are a subset 
			of judgement tests of $\statetwo$ and thus positive by \ih
			
			\item \emph{Type tests.}  Let $n$ be the level of $\ltyctx$ and $q^j$ the type test of $q$ associated to the $j$-th triple of $\DiPref{\arr{\mset{\ldots\ltyctx_i\ldots}}\lintytwo}$. We have three cases, depending on $j$.
			\begin{enumerate}
				\item $j=1$: $\state^1$ is $\tjudg{}{\red{\la\var\ctxp{\var}} }{\ltyctxp{\linty}^\infty_{i\uppt},\arr{\mset{\ldots\ctxhole\ldots}}\lintytwo}$. 
				Notice that $\state^1\tomachbttwo\,\tjudg{} 
				{\blue\var}{\ltyctxp{\linty}^\infty_{i\downpt},\ctxhole}\;=:\; \statethree$.
				Since its context is empty, by definition the state $\statethree$ has no type tests. Furthermore it has the same judgement tests of $\statetwo$ (see \refremark{test-indipendent-k}), which by \ih are positive.
				Also, it respects the focus of $\state^1$. For the moreover part of the J-exhaustible invariant we only have to notice that $\depth{\arr{\mset{\ldots\ctxhole\ldots}}\lintytwo}=1$.
				\item $j$ \emph{is even}: for $q^j$ (with direction $\downpt$) there is a corresponding type test $q'^{j-1}$ with odd index (having direction $\uppt$). We have that $q'^{j-1}\tomachvar q^j$ so we can conclude by \ih and determinism, since $q^j$ must be on the run given by the invariant on $q'^{j-1}$. So it is positive and the moreover condition must be satisfied.
				\item $j\neq 1$ \emph{is odd}: for $q^j$ (with direction $\uppt$) there is a corresponding $q'^{j-1}$ of even index (and direction $\downpt$). We have $q^j\tomachbttwo q'^{j-1}$ so we can conclude by \ih 
				For the moreover part of the invariant we have to notice that $j\geq 3$ so, as noticed in \refremark{type-test-i}, the type context $\ltyctxb$ associated to it must have $\depth \ltyctxb \geq 3$. 
				The transition $\tomachbttwo$ decrease the depth of the type context by just one, so $q^j$ is positive.
				Therefore, the state $q$ is J-exhaustible.
			\end{enumerate}
		\end{itemize}

		\item $\tomachdotthree, \,\tomachdotfour$. Identical to case $\tomachdotone$.
		
		\item Case $\tomacharg$.
		\[\small\begin{array}{cl}
			\statetwo=\infer{\tjudg{}{\tm\tmtwo}{\lintyb}} 
			{\tjudg{}{\blue\tm}{\arr{\mset{\ldots 
							\ltyctxp{\linty^k_{\downpt}}_i\ldots}}{\lintyb}}
				& \tjudgi{}{\tmtwo}{\lintytwo_i (=\ltyctxp{\linty_{\downpt}}_i)}}
			& \tomacharg 
		\end{array}\]
		\[\small\begin{array}{lc}
			\tomacharg &
			\infer{\tjudg{}{\tm\tmtwo}{\lintyb}} 
			{\tjudg{}{\tm}{\arr{\mset{\ldots 
							\lintytwo_i\ldots}}{\lintyb}}
				& \tjudgi{}{\red\tmtwo}{\ltyctxp{\linty^k_{\uppt}}_i}}=\state
		\end{array}\]
		\begin{itemize}
			\item \emph{Judgement tests.} Judgement tests of $\state$ are those of 
			$\statetwo$, which are positive by \ih, plus 
			$\state^\tmtwo \defeq 
			\tjudg{}{\blue\tmtwo}{\ltyctxp{\linty}^\infty_{i\downpt},\ctxhole}$. 
			Please note that 
			$\sigma':\state^\tmtwo\tomachbtone\,\tjudg{}{\red\tm}{\ltyctxp{\linty}^\infty_{i\uppt},
				\arr{\mset{\ldots \ctxhole\ldots}}{\lintyb}}\eqdef\statetwo^t$. 
			Now,
			$\statetwo^\tm$ is a type test of $\statetwo$ and by \ih is 
			positive. Since $\statetwo^\tm$ is positive, it exists a run $\sigma:\statetwo^\tm\toksjam^*\tomachbttwo p$, with $p$ J-exhaustible.
			We can now obtain the following \KSJAM run, obtained by concatenating $\sigma'$ and $\sigma$. 
			\[
				\run:\state^u\tomachbtone\statetwo^t\toksjam^*\tomachbttwo p
			\]
			The moreover part of the J-exhaustible invariant is easily verified, since all the states in $\sigma$ respect the condition for tape tests by \ih and $\depth{\arr{\mset{\ldots \ctxhole\ldots}}{\lintyb}}=1$. We conclude that $\state^\tmtwo$ is positive. 
			
			\item \emph{Type tests.} For each odd type test $q^i$ (with direction $\uppt$) of $q$ there is a corresponding one of $q'$, that is $q'^{i+1}$ (with direction $\downpt)$ and is positive by \ih We have that $q'^{i+1}\tomacharg q^i$, then $q^i$ is positive by determinism of the \KSJAM. For each even type test $q^i$ (with direction $\downpt$) there is a corresponding type test of $q'$, being $q'^{i+1}$ (with direction $\uppt$) and positive by \ih We have $q^i\tomachbtone q'^{i+1}$.
			To prove the moreover part of the statement we need to notice that, by \refremark{type-test-i}, a test with an even index has a type context of even depth. Furthermore, the $\tomachbtone$ transition increases the depth of this context by exactly one. So we conclude that the type context of $q'^{i+1}$ has depth greater then $1$.
		\end{itemize}
		
		\item Case $\tomachjump$.
			\[\begin{array}{cl}
			\statetwo=\infer{\tjudg{}{\tm\tmtwo}{\lintytwo}} 
			{\tjudg{}{\tm}{\arr{\mset{\myldots 
				\linty_i\myldots}}{\lintytwo}}
			& p \defeq\, \tjudgi{}{\blue\tmtwo}{\ltyctxp{\linty^0_\downpt}_i (=\linty_i)}}
			& \tomachjump 
		\end{array}\]

		\[\begin{array}{lc}
			 \tomachjump &
			\jump{p} =\, \infer{\tjudg{}{\blue\var}{\ltyctxp{\initty^0_\downpt}_i (=\linty_i)}}{}
			=\state
		\end{array}\]
			\begin{itemize}
				\item \emph{Judgement tests.} By $\ih$, $\statetwo$ is J-exhaustible and $\statetwo_\tmtwo:=\;\tjudg{}{\blue\tmtwo}{\ltyctxp{\linty}^\infty_{i\downpt},\ctxhole}$ is a judgement test of $\statetwo$, so $\statetwo_\tmtwo$ is positive.
				It exists a run $\sigma:\state'_\tmtwo\tomachbtone\toksjam^*\tomachbttwo s\;:=\; \tjudg{}{\blue\var}{\ltyctxp{\linty}^\infty_{i\downpt},\ctxhole}$, with $s$ a J-exhaustible state and $\sigma$ is the shortest such run.
				We can then lift this run by \reflemma{type-context-lifting} and obtain a \KSJAM run $\sigma_l:\;\statetwo_l:=\,\tjudg{}{\blue\tmtwo}{\ltyctxp{\linty^\infty_\downpt}_{i}}\tomachbtone\toksjam^*\tomachbttwo \tjudg{}{\blue\var}{\ltyctxp{\linty^\infty_\downpt}_{i}}=:s_l$.
				Clearly, the two states $s$ and $s_l$ have the same judgement tests, who are positive by \ih.
				On the other hand, by definition, $\jump p =q\;:=\;\tjudg{}{\blue\var}{\ltyctxp{\linty^0_\downpt}_{i}}$ with $q$ the first encountered axiom state such that $\rho:\;\tjudg{}{\blue\tmtwo}{\ltyctxp{\linty^\infty_\downpt}_{i}}\toksjam ^+ \tjudg{}{\blue\var}{\ltyctxp{\linty^\infty_\downpt}_{i}}$.
				Since both $\sigma_l$ and $\rho$ start from the same state and $\tjudg{}{\blue\var}{\ltyctxp{\linty^\infty_\downpt}_{i}}$ is their first encountered such state, by determinism of the \KSJAM we have that $\jump p$ is well defined in this case.
				The moreover part of the statement comes from the fact that we constructed $\sigma_l$ starting from $\sigma$ that satisfies the condition by \ih The operation of lifting with which we obtained $\sigma_l$ can obviously only increase the depth of the type contexts of the states appearing in $\sigma$.\\

				%

				\item  \emph{Tape tests.} We prove at first the case of odd type tests. We only give the example of the case $i=1$ to help readability, the general case is straightforward generalization. We have two tape tests for states $q'$ and $q$, respectively $\statetwo^1:=\;\tjudg{}{\red\tmtwo}{\ltyctxtwop{\linty}^\infty_{i\uppt},\arr{\mset{\ldots\ctxhole\ldots}}{\lintytwo}}$ and $\state^1\;:=\;\tjudg{}{\red\var}{\ltyctxtwop{\linty}^\infty_{i\uppt},\arr{\mset{\ldots\ctxhole\ldots}}{\lintytwo}}$.
				By $\ih$ $\statetwo$ is positive, so there is a run $\sigma:\statetwo^1\toksjam^*s'\tomachbttwo s$ with $s$ J-exhaustible and respecting the focus of $\statetwo^1$.
				We now want to show that it exists a run $\rho:\state^1\toksjam^+\statetwo^1$, so that by prefixing $\rho$ to $\sigma$ we get a run that exhaust $\state^1$.
				Notice that $s$ also respect the focus of $\state^1$.
				
				Consider the judgement test of $\statetwo$ that we call $\statetwo_u:=\tjudg{}{\blue\tmtwo}{\ltyctxp{\linty}^\infty_{i\downpt},\ctxhole}$. By $\ih$ $\statetwo_u$ is positive, so there is a run $\gamma:\statetwo_u\tomachbtone\toksjam^*\tomachbttwo s_u:=\tjudg{}{\blue\var}{\ltyctxp{\linty}^\infty_{i\downpt},\ctxhole}$. 
				As we said in \refremark{reversibility}, when the backtrack parameter of a state is $\infty$ the \KSJAM run starting from it is reversible. We can then get a run 
				\[
					\gamma^\bot:\;\;(s_u)^\bot\tomachvar s_1
					\toksjam^*s_2\tomacharg
					(\statetwo_u)^\bot
				\] 
				We call $\cdot^\bot$ is the operation that inverts the direction of \KSJAM states.
				We can now apply \reflemma{type-context-lifting} to this run, since $\ltyctx=\arr{\mset{\ldots\ltyctxtwop{\linty}\ldots}}\lintytwo$, ending up with the desired $\rho$. Explicitly
				\begin{align*}
					\rho:\;\;s_u^l:=\;\tjudg{}{\red\var}{\ltyctxtwop{\linty}^\infty_{i\uppt},\arr{\mset{\ldots\ctxhole\ldots}}\lintytwo}\tomachvar s_{1}^l
					\toksjam^*s_{2}^l
					\\\tomacharg
					\tjudg{}{\red\tmtwo}{\ltyctxtwop{\linty}^\infty_{i\uppt},\arr{\mset{\ldots\ctxhole\ldots}}{\lintytwo}}:={q'}^{l}_u
				\end{align*}
				Since, following the same argument as in the previous \emph{Judgement test} case, we have that $\state$ can be obtained by lifting the state $s_u$, therefore we have that $s_u^l$ is exactly $\state^1$. By prefixing the run $\rho$ to $\sigma$ we then obtain the run $\tau:\state^1\toksjam^*s$. 

				We now have to check the moreover part of the J-exhaustible invariant. By \ih, on the run $\gamma$ the condition for runs corresponding to a judgement test is verified. We define $\gamma^\bot_1$ as the sub-run of $\gamma^\bot$ such that $\gamma^\bot_1 :s_1\toksjam^*s_2$. By reversibility of the \KSJAM in this case, we can say that each state $(\tyd, \ruleoc, (\ltyctxp{\linty},\ltyctxb), \pol,\infty)$ in $\gamma^\bot_1$ is such that $\depth \ltyctxb \geq 1$. Even more, this is true in the sub-run of $\tau$ that corresponds to the lifting of $\gamma^\bot_1$, namely $\tau_l:s_{1}^l\toksjam^*s_{2}^l$. Also, both $s_u^l$ and ${q'}^{l}_u$ have clearly a type context of depth equal to $1$.
				Finally, by \ih, on $\sigma$ is verified the condition for runs corresponding to tape tests. 
				We can then concude that for each state in the run $\tau$ the condition for tape tests is verified, so $\state^1$ is positive.\\

				We now check the case of even type tests. 
				Let $\statetwo_l\defeq \tjudg{}{\blue\tmtwo}{\ltyctxp{\linty}^\infty_{i\downpt},\ctxhole}$ be the judgement test of $\statetwo$. Since is positive by \ih we have that it exists a run $\sigma: \statetwo_l\toksjam^* \tjudg{}{\blue\var}{\ltyctxp{\linty}^\infty_{i\downpt},\ctxhole}$. Let now $\statetwo_i\defeq \tjudg{}{\blue\tmtwo}{\ltyctx''\ctxholep{\linty}^\infty_{i\downpt},\ltyctx'}$ be an even tape test of $\statetwo$, with $\ltyctx=\ltyctx''\ctxholep{\ltyctx'}$ as in the definition of tape tests. By \reflemma{type-context-lifting} we can lift $\sigma$ and obtain a run $\rho:\statetwo_i\toksjam^*\tjudg{}{\blue\var}{\ltyctx''\ctxholep{\linty}^\infty_{i\downpt},\ltyctx'}\eqdef \state_i$. This $\state_i$ is exactly the tape test of $\state$ that corresponds to $\statetwo_i$ and is positive by determinism of the \KSJAM. The moreover part of the statement holds for $\sigma$ by \ih, therefore it also hold for $\rho$.
			\end{itemize}
		\item Case $\tomachbtone$.
			\[\begin{array}{cl}
		\statetwo=\infer{\tjudg{}{\tm\tmtwo}{\lintytwo}} 
		{\tjudg{}{\tm}{\arr{\mset{\ldots 
						\linty_i\ldots}}{\lintytwo}}
			& 
			\tjudgi{}{\blue\tmtwo}{\ltyctxp{\lintyb^{k+1}_{\downpt}}_i(=\linty_i)}}
		& \tomachbtone 
		\end{array}\]

		\[\begin{array}{lc}
		\tomachbtone &
		\infer{\tjudg{}{\tm\tmtwo}{\lintytwo}} 
		{\tjudg{}{\red\tm}{\arr{\mset{\ldots 
						\ltyctxp{\lintyb^k_{\uppt}}_i\ldots}}{\lintytwo}}
			& \tjudgi{}{\tmtwo}{\linty_i}}=\state
		\end{array}\]

			\begin{itemize}
				\item \emph{Judgement tests.} Note that $\state$ has the same judgement tests of $\statetwo$, which are 
			positive by the \ih
			\item \emph{Type tests.} We call $q^1$ the first type test of $q$. We have $q^1\defeq\, \tjudg{}{\red \tm}{\ltyctxp{\lintyb}^\infty_{i\uppt}, \arr{\mset{\myldots \ctxhole
						\myldots}}{\lintytwo}}$.
			Please note that $q'^{u}\defeq \tjudg{}{\blue u}{\ltyctxp{\lintyb}^\infty_{i\downpt}}, \ctxhole$ is a judgement test of $q'$ such that $q'^{u}\tomachbtone q^1$. By \ih and by determinism of the \KSJAM then $q^1$ is positive. We now look at all the others type tests. For each odd type test $q^i$ of $q$ (with direction $\uppt$), there is a corresponding type test $q'^{i-1}$ of $q'$ (with direction $\downpt$) that is positive by \ih and such that $q'^{i-1}\tomachbtone q^i$. So $q^i$ is positive by determinism of the \KSJAM. For each even type test $q^i$ (with direction $\downpt$) of $q$, there is a corresponding odd type test $q'^{i-1}$ of $q'$ (with direction $\uppt$) that is also positive by \ih, such that $q^i\tomacharg q'^{i-1}$. We have that $i\geq 2$ and the transition $\tomacharg$ decreases the depth of the type context of exactly one, so type type context of $q'^{i-1}$ is greater then $1$.  

			\end{itemize}
		\item Case $\tomachbttwo$.
			\[\small
			\statetwo\defeq\infer*{\infer{\tjudg{}{\red{\la\var\ctxp{\var}}} 
			{\arr{\mset{\myldots\ltyctxp{\lintyb^k_\uppt}_i\myldots}}\lintytwo}}{}}
	{\infer[i]{\tjudg{}{\var}{\linty_i (=\ltyctxp{\lintyb}_i)}}{}}
	  \tomachbttwo 
	 \infer*{\infer{\tjudg{}{\la\var\ctxp{\var}} 
			{\arr{\mset{\myldots\linty_i\myldots}}\lintytwo}}{}}
	{\infer[i]{\tjudg{}{\blue\var}{\ltyctxp{\lintyb^{k+1}_\downpt}_i}}{}}\eqdef \state\]
			\begin{itemize}
				\item \emph{Judgement tests.} The first type test of $q'$ is $q'^1 \defeq\, \tjudg{}{\red{\la\var\ctxp{\var}}}{\ltyctxp{\lintyb}^\infty_{i\uppt}, \arr{\mset{\myldots \ctxhole
						\myldots}}{\lintytwo}}$.
			Note that $q'^1\tomachbttwo \tjudg{}{\blue \var}{\ltyctxp{\lintyb}^\infty_{i\downpt}}, \ctxhole\,\eqdef q''$ and that $q''$ exhaust $q'^1$ and is the first such state. Since $q'^1$ is positive, then $q''$ is J-exhaustible. Note that $q''$ has the same judgement tests of $q$, which are then positive. For the moreover part of the statement we just have to notice that $\depth{\arr{\mset{\myldots \ctxhole
						\myldots}}{\lintytwo}}=1$.
			\item \emph{Type tests.} For each even type test $q^i$ (with direction $\downpt$) of $q$, there is a corresponding odd type test $q'^{i+1}$ of $q'$ (with direction $\uppt$) that is positive and such that $q'^{i+1}\tomachbttwo q^i$. So $q^i$ is positive by determinism of the \KSJAM. For each odd type test $q^i$ of $q$ (with direction $\uppt$), there is a corresponding type test $q'^{i+1}$ of $q'$ (with direction $\downpt$) that is positive by \ih and such that $q^i\tomachvar q'^{i+1}$. For the moreover part we notice that $q^i$ has a type context of depth $\geq 1$ and the $\tomachvar$ transition increases this depth by exactly $1$.
			\end{itemize}

	\end{itemize}
\end{proof}


\cjumpwelldef*
\begin{proof}
	This is a consequence of the \emph{Judgement test} part of the $\tomachjump$ case in the proof of \reflemma{K-invariant-ksjam-str}.
\end{proof}

\section{Proof of the \KJAM/\KSJAM bisimulation}
\label{sect:apx5}

\subparagraph{Relating Logs and Tapes with Typed Position.}

Also this section contains proofs which are rephrasing of proofs appearing in~\cite{POPL2021}. We added them for the sake of completeness.
The \KJAM uses the log $L=l_1,\dots,l_n$ to represent the arguments $u_1,\dots,u_n$ in which the position of the current state is contained. Each $u_i$ corresponds to the answer of a query for an argument associated with the variable in the logged position $l_i$. In the \KSJAM, this information is given by the sub-derivations for $u_1,\dots,u_n$ in which the judgement of the current state occurs.

Notice that in the \KJAM, the transitions $\tomacharg$ and $\tomachbtone$ respectively add and remove exactly one logged position from the log. In the \KSJAM, these correspond to the only transitions that enter or exit arguments in the derivation. Also, the transitions $\tomachvar$, $\tomachbttwo$ and $\tomachjump$ in the \KJAM are responsible for creating and reading the content of a logged position. In the \KSJAM, this translates into exiting or entering an axiom typing a variable.

The tape of the \KJAM is used to store information about the logged position for which the machine is currently searching an associated argument, or is backtracking to. In the \KSJAM, this information is encoded in the type context of the state. The \KJAM adds logged positions to the tape through the transitions $\tomachvar$ and $\tomachbtone$, and removes them through $\tomacharg$ and $\tomachbttwo$. In the \KSJAM, these correspond to transitions that respectively increase or decrease the depth of the state’s type context.
Notice also that after a $\tomachjump$ transition the tape, and thus the type context, remains unchanged. By examining \reffig{LJAM-strat} and \reffig{sjam-strat}, it becomes evident that there is a logged position on the \KJAM tape for every type sequence $S$ in which it lies the hole $\ctxhole$ of the current type context $\ltyctx$ of the \KSJAM.
These are the main intuitions behind the \emph{extraction} procedure, which allows us to retrive a \KJAM state $\estate \state$ from an exhaustible \KSJAM state $\state$. 

From J-exhaustible states one is able to \emph{extract} \KJAM states, as the 
following definition shows. Please note that the definition is well-founded, 
precisely because the objects are J-exhaustible states. Indeed, the induction 
principle used to define J-exhaustability allows recursive definition on 
J-exhaustible states to be well-behaved. 
\begin{remark}
	\label{rem:reversibility}
	We first make a small remark about how logged positions are managed in the \KJAM. In the transition $\tomachvar$, the logged position that is created contains a \emph{global} context with respect to the term: if $t$ is the initial term, we have $t = \ctx\ctxholep{\lambda x.\ctxtwo_{n}\ctxholep{x}}$. In the $\tomachjump$ transition, the machine fully restores the context $\ctxtwo$ contained in the logged position $(x,\ctxtwo,L')$. This mechanism allows the \KJAM to avoid backtracking runs. In the literature, the \IAM usually uses \emph{local} logged positions: in the $\tomachvar$ transition, only the context $\lambda x.\ctxtwo_n$ would be saved.
	To avoid unnecessary complexity, the \KJAM logs positions globally. However, note that in the $\tomachbttwo$ transition, only the local information about the context $\lambda x.\ctxtwo_n$ contained in $l'$ is actually used, since the information about the global context $C$ is already present in the state.
\end{remark}
\begin{definition}[Extraction of logged positions]
	Let $\state$ be an J-exhaustible \KSJAM state in a derivation $\tyd$, $\tm$ 
	be the final term in $\tyd$, and $\state_\focus$ be a judgment or type 
	test of $\state$. Since 
	$\state$ is 
	J-exhaustible, there is an exhausting run 
	$\state_\focus\toksjam^+ \statetwo
	\in\exstates_J$. Let $\var$ be the variable of $\statetwo$. Then the logged position extracted from 
	$\state_\focus$ is 
	$\elpos{\state_\focus} \defeq 
	(\var,\ctx,\elpos{\statetwo^1}\cdot\ldots\cdot\elpos{\statetwo^n})$, where $\ctx$ is the only context such that $t=\ctxp{x}$. Finally, $\statetwo^i$ is the 
	$i$-th judgment test of $\statetwo$.
\end{definition}

\begin{definition}[Extraction of logs, tapes, and states]
	\label{def:extraction}
	Let $\state=(\tyd, \ruleoc, (\linty,\ltyctx), \pol,k)$ be an J-exhaustible 
	\KSJAM 
	state where $\tm$ is the final term in 
	$\tyd$, and $\ruleoc$ is $\tjudg{\tye}{\tmtwo}{\ltyctxp\linty}$. The \KJAM 
	state extracted from $\state$  is 
	$\estate{\state} \defeq \nopolkstate{\tmtwo}{\ctx_\state}{\etape\state}{\elog{\state}}{\pol}{k}$ 
	where
	\begin{itemize}
		\item \emph{Context}: $\ctx_\state$ is the only term context such that $\tm = \ctx_\state\ctxholep\tmtwo$;
		\item \emph{Log}:
		$\elog{\state}\defeq\lpos_1\cdots\lpos_i\cdots\lpos_n$ where $\lpos_i = \elpos{\state^i_\focus}$ where 
		$\state^i_\focus$ is the $i$-th judgment test of $\state$.		
		\item \emph{Tape}: $\etape\state = \etapeauxs{\ltyctx,0}$ where 
		$\etapeauxs{\ltyctx,i}$ is the auxiliary function 
		defined by induction on $\ltyctx$ as follows.
		\[\begin{array}{lcl}
			\etapeauxs{\ctxhole,i} &\defeq &\stempty 
			\\
			\etapeauxs{\arr\mty{\ltyctx}, i} & \defeq & \resm\cdot 
			\etapeauxs{\ltyctx,i}
			\\
			\etapeauxs{\arr{\mset{\myldots\ltyctx\myldots}}\lintytwo,i}
			& \defeq & \elpos{\state^i_\focus}\cdot\etapeauxs{\ltyctx, i+1}
		\end{array}
		\]
		where $\state^i_\focus$ is the $i$-th type test of $\state$.
	\end{itemize}
	We use $\bisimtypes$ for the extraction relation between J-exhaustible \KSJAM states and \KJAM states defined as $(\state, \estate{\state}) \in\bisimtypes$.
\end{definition}

First of all, we show that the extracted state respects the \KJAM invariant about the length of the log.

\begin{lemma}
	\label{l:extraction-length}
	Let $\state$ be an J-exhaustible \KSJAM state and $\estate{\state} = 
	\nopolkstate{\tm}{\ctx_\state}{\etape\state}{\elog{\state}}{\pol}{k}$ the \KJAM state extracted from it. Then the level of 
	$\ctx_\state$ is exactly the length of $\elog{\state}$, that is, $(\tm, \ctx_\state,\elog{\state})$ is a logged 
	position.
\end{lemma}

\begin{proof}
	The length of $\elog{\state}$ is the number of judgment tests of $\state$, 
	which is the number of 
	$\tyapp$ rules traversed descending from the focused judgment $\ruleoc$ of 
	$\state$ to the final judgment of $\tyd$. 
	The level of $\ctx_\state$ is the number of arguments in which the hole of $\ctx_\state$ is contained, which are 
	exactly the number of 
	$\tyapp$ rules traversed descending from $\ruleoc$ to the final judgment of the current derivation
	$\tyd$.
\end{proof}

\begin{proposition}[\KJAM-\KSJAM bisimulation]
	\label{prop:bisim-sjam-ljam}
	Let $\tm$ be a closed and $\towh$-normalizing term, and 
	$\tyd\pof\tjudg{}{\tm}{\initty}$ be a type derivation. Then,  
	$\bisimtypes$ is a strong bisimulation between J-exhaustible \KSJAM states on $\tyd$ and \KJAM states on $\tm$. 
	Moreover, if $\state_\tyd \bisimtypes \state_\l$ then $\state_\tyd$ is \KSJAM reachable if and only if $\state_\l$ is 
	\KJAM reachable.
\end{proposition}

\begin{proof}
	Assuming the bisimulation part of the statement, the moreover part follows from a trivial induction on the length of 
	the initial run, since initial state are bisimilar and the bisimulation is exactly the fact that $\bisimtypes$ is 
	stable by transitions.
	
	For the bisimulation part, we consider each possible transitions. We focus on the half of the proof showing that \KSJAM 
	transitions are simulated by the \KJAM, the other half is essentially identical. 

	\begin{itemize}
		\item Case $\tomachdotone$.
		\[\begin{array}{clc}
			\statetwo=\infer{\tjudg{}{\red{\tm\tmtwo}}{\ltyctxp{\initty^k_{\uppt}} 
					(=\linty)}} 
			{\tjudg{}{\tm}{\arr{\mty}{\linty}} & \mset\vdash} &
			\tomachdotone &
			\infer{\tjudg{}{\tm\tmtwo}{\linty 
			}}{\tjudg{}{\red\tm}{\arr{\mty}{\ltyctxp{\initty^k_{\uppt}}}} & 
				\mset\vdash}=\state
			\\[8pt]	
			\bisimtypes&&
			\\[8pt]
			\estate\state=\dkstate{ \tm\tmtwo }{ \ctx_\statetwo }{ \etape{\statetwo} }{ \elog{\statetwo} }{k} 
			&\iamdap& 
			\dkstate{ \tm }{ \ctx_{\statetwo}\ctxholep{\ctxhole\tmtwo} }{ \resm\cdot \etape{\statetwo} }{ 
				\elog{\statetwo} }{k} = \state_\l
		\end{array}\]
		
		Note that $\ctx_\state = \ctx_\statetwo\ctxholep{\ctxhole\tmtwo}$, $\elog\state = 
		\elog{\statetwo}$, and $\etape{\state} = \resm\cdot\etape{\statetwo}$. Then, $\state_\l = \estate{\state}$, that is, 
		$\state \bisimtypes \state_\l$.
		
		\item Case $\tomachdottwo$. Identical to the previous one.
		
		\item Case $\tomachvar$.
		\[\begin{array}{cl}
			\statetwo=\infer*{\infer{\tjudg{}{\la\var\ctxtwo_n\ctxholep{\var}} 
					{\arr{\mset{\ldots\linty_i\ldots}}\lintytwo}}{}}
			{\infer[i]{\tjudg{}{\red\var}{\ltyctxp{\initty^k_{\uppt}}_i 
						(=\linty_i)}}{}}  
			& 
			\tomachvar	
			\\[8pt]
			\bisimtypes&
			\\[8pt]
			\estate\state=\dkstate{ \var }{ \underbrace{\ctxp{\l\var.\ctxtwo_n}}_{=\ctx_\statetwo} }{ \etape\statetwo }{ 
				\underbrace{\tlog_n\cdot\tlog}_{=\elog\statetwo} }{k}
			&\tomachvar 
		\end{array}\]
		\[\begin{array}{lc}
			
			\tomachvar
			& \infer*{\infer{\tjudg{}{\blue{\la\var\ctxtwo_n\ctxholep{\var}}}
					{\arr{\mset{\ldots\ltyctxp{\initty^k_{\downpt}}_i\ldots}}\lintytwo}}{}}
			{\infer[i]{\tjudg{}{\var}{\linty_i}}{}}=\state		
			\\[8pt]
			&
			\\[8pt]
			\tomachvar &
			\ukstate{ \l\var.\ctxtwo_n\ctxholep\var}{ \ctx }{ 
				(\var,\ctxp{\l\var.\ctxtwo_n},\tlog_n\cdot\tlog)\cdot\etape\statetwo }{ \tlog }{k} = \state_\l		
		\end{array}\]
		First of all, in this case $\ctx_\statetwo$ has shape $\ctxp{\l\var.\ctxtwo_n}$ for some $n$, and the descending path from the 
		focused judgement to the final judgement passes through the showed $\tylam$ rule. Then $\ctx_{\state} = \ctx$.
		
		About the log, 
		by \reflemma{extraction-length} there is a correspondance between the level of term contexts and the length of the 
		extracted log, so that $\elog\statetwo$ has at least length $n$, that is, $\elog\statetwo = \tlog_n\cdot\tlog$, and 
		$\elog\state = \tlog$.
		
		About the tape, note that $\etape\state = 
		\elpos{\state^1_\focus}\cons\etapeaux{\ltyctx,1}\state$ where 
		$\state^1_\focus$ is the first type test of $\state$. To show that $\etape\state = 
		(\var,\ctxp{\l\var.\ctxtwo_n},\tlog_n\cdot\tlog)\cdot\etape\statetwo$ we have to show two things:
		\begin{enumerate}
			\item $\elpos{\state^1_\focus} = 
			(\var,\ctxp{\l\var.\ctxtwo_n},\tlog_n\cdot\tlog)$.
			Note that $\state^1_\focus$ is $\tjudg{}{\red{\la\var\ctxtwo_n\ctxholep{\var}} }
			{\ltyctxp{\initty^\infty_\uppt}_{i},\arr{\mset{\ldots\ctxhole\ldots}}\lintytwo}$.
			Note also that 
			$\state^1_\focus\tomachbttwo\,\tjudg{} 
			{\blue\var}{\ltyctxp{\initty^\infty_\downpt}_{i},\ctxhole} = \statethree$, 
			where $\statethree$ focusses on the same 
			judgement of $\statetwo$, and that $\statethree$ is the state that J-exhausts $\state^1_\focus$. By definition of 
			extraction, $\elpos{\state^1_\focus} = (\var,\ctxp{\l\var.\ctxtwo_n},\tlog_n\cdot\tlog)$.
			
			\item $\etapeaux{\ltyctx,1}\state = \etape\statetwo$, that is, 
			$\etapeaux{\ltyctx,1}\state = 
			\etapeaux{\ltyctx,0}\statetwo$. Note that $\etapeaux{\ltyctx,1}\state$ and 
			$\etapeaux{\ltyctx,0}\statetwo$ may differ only 
			in the content of logged positions (obtained by extracting from tape tests), which is the only thing that depends on 
			the direction and the state, the rest being uniquely determined by the type 
			context $\ltyctx$. Here one has to repeat the reasoning done in the $\tomachbttwo$ case of the proof of the J-exhaustible invariant (\reflemma{K-invariant-ksjam-str}), that shows that the tape test of index $i>1$ for $q$ and the one of index $i-1$ of $q'$ exhaust on the same state, and thus induce the same logged positions. Then $\etapeaux{\ltyctx,1}q =\etape{q'}$.
		\end{enumerate}
		Then $\etape\state = (\var,\ctxp{\l\var.\ctxtwo_n},\tlog_n\cdot\tlog)\cdot\etape\statetwo$, and so $\state_\l = \estate{\state}$, that 
		is, $\state \bisimtypes \state_\l$.

		\item Cases $\tomachdotthree$ and $\tomachdotfour$. They are identical to case 
		$\tomachdotone$.
		
		\item Case $\tomacharg$.
		\[\small\begin{array}{cl}
			\statetwo=\infer{\tjudg{}{\tm\tmtwo}{\linty}} 
			{\tjudg{}{\blue\tm}{\arr{\mset{\ldots 
							\ltyctxp{\initty^k_{\downpt}}_i\ldots}}{\linty}}
				& \tjudgi{}{\tmtwo}{\lintytwo_i (=\ltyctxp{\initty}_i)}}
			& \tomacharg 
			\\[8pt]
			\bisimtypes&
			\\[8pt]
			\estate\statetwo=\ukstate{ \tm }{ \underbrace{\ctxtwop{\ctxhole\tmtwo}}_{=\ctx_{\statetwo}} }{ 
				\underbrace{\elpos{\statetwo^1_\focus}\cdot\etapeaux{\ltyctx,1}{\statetwo}}_{=\etape\statetwo}
			}{ \elog\statetwo }{k}
			& \tomacharg 
		\end{array}\]
		\[\small\begin{array}{lc}
			 \tomacharg &
			\infer{\tjudg{}{\tm\tmtwo}{\linty}} 
			{\tjudg{}{\tm}{\arr{\mset{\ldots 
							\lintytwo_i\ldots}}{\linty}}
				& \tjudgi{}{\red\tmtwo}{\ltyctxp{\initty^k_{\uppt}}_i}}=\state
			\\[8pt]
			&
			\\[8pt]
			\tomacharg &
			\dkstate{ \tmtwo }{ \ctxtwop{\tm\ctxhole} }{ 
				\etapeaux{\ltyctx,1}{\statetwo} }{ 
				\elpos{\statetwo^1_\focus}\cdot\elog\statetwo }{k} = \state_\l
		\end{array}\]
		where $\statetwo^1_\focus$ is the first type test of $\statetwo$. Obviously, $\ctx_{\state} = \ctxtwop{\tm\ctxhole} 
		$. For the log we have to show that $\elog\state$ is equal to $\elpos{\statetwo^1_\focus}\cdot\elog\statetwo $, which 
		amounts to show that the first judgement test $\state^1$ of $\state$ exhausts on the same state as the first tape test 
		$\statetwo^1_\focus$ of $\statetwo$. This is exactly the reasoning done in the proof of the J-exhaustible invariant. 
		Similarly, one obtains that $\etapeaux{\ltyctx,1}{\statetwo}=\etape{q}=\etapeaux{\ltyctx,0}{q}$.

		\item Case $\tomachjump$.
		
		\[\begin{array}{cl}
			\statetwo=\infer{\tjudg{}{\tm\tmtwo}{\linty}} 
			{\tjudg{}{\tm}{\arr{\mset{\ldots 
							\lintytwo_i\ldots}}{\linty}}
				& 
				p\;:=\;\tjudgi{}{\blue\tmtwo}{\ltyctxp{\initty^0_{\downpt}}_i(=\lintytwo_i)}}
			& \tomachjump 
			\\[8pt]
			\bisimtypes&
			\\[8pt]
			\estate\statetwo=\ukstate{ \tmtwo }{ \ctxtwop{\tm\ctxhole}}{ \etape\statetwo }{ 
				\underbrace{(\var, \ctx,\tlog')\cdot\tlog}_{=\elog\statetwo} }{0}
			& \tomachjump 
		\end{array}\]

		\[\begin{array}{lc}
			 \tomachjump &
			\jump{p} =\, \infer{\tjudg{}{\blue\var}{\ltyctxp{\initty^0_\downpt}_i (=\linty_i)}}{}
			=\state
			\\[8pt]
			&
			\\[8pt]
			\tomachjump &
			\ukstate{ \var }{ C }{ \etape\state }{ \tlog' }{0}
			= \state_\l
		\end{array}\]
		
		The first element of $\elog{\statetwo}$ is $\elpos{\statetwo^1_j}=(x,\ctx,\tlog')$, with $\statetwo^1_j$ being the first judgement test of $\statetwo$. Since $\statetwo$ is a reachable state, it exists a run $\statetwo^1_j\tosiam^*s$ and $\elpos{\statetwo^1_j}$ is the logged position extracted from $s$. By the same reasoning made in the case \emph{Judgement Tests} of $\tomachjump$ in \reflemma{K-invariant-ksjam-str}, we have that the two states $s$ and $\state$ are the same. Hence, the logged position extracted from $\state$, \ie the first three components of the state $\state_\l$, are the same as in $\elpos{\statetwo^1_j}$.
		
		Now we have to prove that $\etape{\statetwo}=\etape{\state}$. Since everything except the logged positions depends on the shape of the type context $\ltyctx$ we only have to show that the logged positions on the two tapes are the same. Also, since the extraction of a logged position from a test of the \KSJAM depends only on its exhausting state, we have to show that for each $i\leq depth(\ltyctx)$ the type tests $q^i$ of $q$ and $q'^i$ of $q'$ exhausts on the same state. This is what has been done in the \emph{Type test} case of \reflemma{K-invariant-ksjam-str}.
		

		\item Case $\tomachbtone$.
		\[\begin{array}{cl}
		\statetwo=\infer{\tjudg{}{\tm\tmtwo}{\linty}} 
		{\tjudg{}{\tm}{\arr{\mset{\ldots 
						\lintytwo_i\ldots}}{\linty}}
			& 
			\tjudgi{}{\blue\tmtwo}{\ltyctxp{\initty^{k+1}_{\downpt}}_i(=\lintytwo_i)}}
		& \tomachbtone 
		
			\\[8pt]
			\bisimtypes&
			\\[8pt]
		\estate\statetwo=\ukstate{ \tmtwo }{ \underbrace{\ctxtwop{\tm\ctxhole}}_{=\ctx_{\statetwo} }}{ \etape\statetwo }{ 
\underbrace{\elpos{\state^1_\focus}\cdot\tlog}_{=\elog\statetwo} }{k+1}
		& \tomachbtone
		\end{array}\]

		\[\begin{array}{lc}
		\tomachbtone &
		\infer{\tjudg{}{\tm\tmtwo}{\linty}} 
		{\tjudg{}{\red\tm}{\arr{\mset{\ldots 
						\ltyctxp{\initty^k_{\uppt}}_i\ldots}}{\linty}}
			& \tjudgi{}{\tmtwo}{\lintytwo_i}}=\state
			\\[8pt]
			&
			\\[8pt]
		\tomachbtone &
		\dkstate{ \tm }{ \ctxtwop{\ctxhole\tmtwo} }{ \elpos{\state^1_\focus}\cdot\etape\statetwo }{ \tlog }{k}
= \state_\l
		\end{array}\]
  where $\statetwo^1_\focus$ is the first judgement test of $\statetwo$. Obviously, $\ctx_{\state} = 
\ctxtwop{\ctxhole\tmtwo} 
$. For the log, there is nothing to prove. For the tape, we have to show that $\etape\state$ is equal to 
$\elpos{\state^1_\focus}\cdot\etape\statetwo$, which 
amounts to show two things. First, that the first tape test $\state^1$ of $\state$ exhausts on the same state as 
the first judgement test 
$\statetwo^1_\focus$ of $\statetwo$. Second, that $\etapeaux{\ltyctx,1}\state = 
\etape\statetwo = 
\etapeaux{\ltyctx,0}\statetwo$. Both points follow exactly the reasoning done 
in the proof of the J-exhaustible 
invariant.   
		
\item Case $\tomachbttwo$.
		\[\small\begin{array}{cl}
		\statetwo=\infer*{\infer{\tjudg{}{\red{\la\var\ctxtwo_n\ctxholep{\var}}}
				{\arr{\mset{\ldots\ltyctxp{\initty^k_{\uppt}}_i\ldots}}\lintytwo}}{}}
		{\infer[i]{\tjudg{}{\var}{\linty_i (= \ltyctxp\initty_i)}}{}}
		& \tomachbttwo
		\\[8pt]
		\bisimtypes&
		\\[8pt]
		\estate\statetwo=\dkstate{ \la\var\ctxtwo_n\ctxholep{\var} }{ \ctx_\statetwo }{
			\underbrace{(\var,\ctx_{\statetwo}\ctxholep{\la\var\ctxtwo_n},\tlog_n)\cons\etapeaux{\ltyctx,1}\statetwo}_{=\etape\statetwo}
			 }{ \elog\statetwo }{k}
		&	\tomachbttwo
		\end{array}\]

		\[\small\begin{array}{lc}
		 \tomachbttwo&\infer*{\infer{\tjudg{}{\la\var\ctxtwo_n\ctxholep{\var}} 
				{\arr{\mset{\ldots\linty_i\ldots}}\lintytwo}}{}}
		{\infer[i]{\tjudg{}{\blue\var}{\ltyctxp{\initty^{k+1}_{\downpt}}_i}}{}}=\state
		\\[8pt]
		
		\\[8pt]
		 
		\tomachbttwo&\ukstate{ \var}{ \ctx_{\statetwo}\ctxholep{\la\var\ctxtwo_n} }{ 
		\etapeaux{\ltyctx,1}\statetwo }{	\tlog_n 
\cons\elog\statetwo }{k+1} = \statetwo_\l
		\end{array}\]
		
		About the tape of $\estate\statetwo$, note that $\etape\statetwo = 
\elpos{\statetwo^1_\focus}\cons\etapeaux{\ltyctx,1}\statetwo$ where 
$\state^1_\focus$ is 
the first type test of $\statetwo$. We have to show that $\statetwo^1_\focus$ exhausts on $\var$, so that 
$\elpos{\statetwo^1_\focus} = (\var,\ctx_{\statetwo}\ctxholep{\l\var.\ctxtwo_n},\tlog_n)$ for some $\tlog_n$.
		Note that $\statetwo^1_\focus$ is $\tjudg{}{\red{\la\var\ctxp{\var}} }
			{\ltyctxp{\initty}^\infty_{i\uppt},\arr{\mset{\ldots\ctxhole\ldots}}\lintytwo}$.
		Moreover, 
$\statetwo^1_\focus\tomachbttwo\,\tjudg{} 
			{\blue\var}{\ltyctxp{\initty}^\infty_{i\downpt},\ctxhole} = \statethree$, 
			where $\statethree$ focusses on the same 
judgement of $\state$, and that $\statethree$ is the state that J-exhausts $\statetwo^1_\focus$. By definition of 
extraction, $\elpos{\statetwo^1_\focus} = (\var,\ctx_{\statetwo}\ctxholep{\l\var.\ctxtwo_n},\tlog_n)$ where $\tlog_n$ is the extraction of the first 
$n$ judgement tests of $\state$. Notice how the $\ctx_{\statetwo}$ in $\elpos{\statetwo^1_\focus}$ is the same as the context of $\statetwo$ since it is uniquely defined by the structure of the term. Then $\ctx_\state = \ctx_{\statetwo}\ctxholep{\la\var\ctxtwo_n}$ and $\elog\state = 
\tlog_n \cons\elog\statetwo$.

About the tape, for $\state$ we have to prove that 
$\etapeaux{\ltyctx,1}\statetwo = \etape\state = 
\etapeaux{\ltyctx,0}\state$. This is done as for $\tomachvar$, mimicking the 
reasoning in the proof of the J-exhaustible 
invariant. Then, $\state_\l = \estate{\state}$, that is, $\state \bisimtypes \state_\l$.
	\end{itemize}
\end{proof}

\section{Proofs of Section 5}

\propdepthcontextksjam*
\begin{proof}
	By induction on the length of the reduction from the initial state to $\state$. Let $\rho:\;\;\tjudg{}{\red u}{\ctxholep{\initty^k_{\uppt}}}\toksjam^n s$ an initial run. If $|\rho|=0$ then the result trivially holds. Now let us consider $\rho:\;\;\tjudg{}{\red u}{\ctxholep{\initty^k_{\uppt}}}\toksjam^{n-1}s'\toksjam s$. We have different cases:
	\begin{itemize}
		\item Case $\tomachdotone$. Then $s'\defeq\, \tjudg{}{\red{\tm\tmtwo}}{\ltyctxp{\initty^p_{\uppt}}}$ and $s\defeq \tjudg{}{\red{\tm}}{\arr{\mty}\ltyctxp{\initty^p_{\uppt}}}$. So, $\depth{\arr{\mty}\ltyctx} = \depth{\ltyctx} =_{\ih} 2(\kparam - p)$.
		\item Case $\tomachdottwo$. Then $s'\defeq\, \tjudg{}{\red{\la\var\tm}}{\arr{\mty}{\ltyctxp{\initty^p_{\uppt}}}}$ and $s\defeq \tjudg{}{\red{\tm}}{\ltyctxp{\initty^p_{\uppt}}}$. So, $\depth{\ltyctx} = \depth{\arr{\mty}\ltyctx} =_{\ih}  2(\kparam - p)$.
		\item Cases $\tomachdotthree$ and $\tomachdotfour$. Identical to the previous ones, by \ih we get that $\depth{\ltyctx} = \depth{\arr{\mty}\ltyctx} =_{\ih}  2(\kparam - p)+1$.
		\item Case $\tomachvar$. Then $s'\defeq\, \tjudg{}{\red\var}{\ltyctxp{\initty^p_{\uppt}}}$ and $s:=\tjudg{}{\blue{\la\var\ctxp{\var}}}{\arr{\mset{\ldots\ltyctxp{\initty^p_{\downpt}}\ldots}}\lintytwo}$. So, $\depth{\arr{\mset{\ldots\ltyctx\ldots}}\lintytwo}=1+ \depth{\ltyctx} =_{\ih} 1 + 2(\kparam - p)$.	
		\item Case $\tomacharg$. Then $s':=\tjudg{}{\blue{\tm}}{\arr{\mset{\ldots\ltyctxp{\initty^p_{\downpt}}\ldots}}\lintytwo}$ and $s\defeq\, \tjudg{}{\red\var}{\ltyctxp{\initty^p_{\uppt}}}$. So, $\depth{\ltyctx}= \depth{\arr{\mset{\ldots\ltyctx\ldots}} \lintytwo} -1 =_{\ih} 2(\kparam -p) + 1 -1$.
		\item Case $\tomachjump$. The type context, as well as the direction, do not change from $s'$ to $s$.  
		\item Case $\tomachbtone$. We have $s'\defeq\,\tjudgi{}{\blue\tmtwo}{\ltyctxp{\initty^{p+1}_\downpt}_i (=\linty_i)}$ and $s\defeq \, \tjudg{}{\red\tm}{\arr{\mset{\myldots \ltyctxp{\initty^p_\uppt}_i\myldots}}{\lintytwo}}$. By \ih on $s'$ we have that $\depth{\linty}=2(\kparam - (p+1)) +1$. So, $\depth{\arr{\mset{\myldots \ltyctxp{\initty_\uppt}_i\myldots}}{\lintytwo}}=\depth{\linty}+1=_{\ih} 2(\kparam - (p+1)) +1 + 1 = 2(\kparam - p)$.
		\item Case $\tomachbttwo$ is identical to the previous one. 
	\end{itemize}
\end{proof}

\simuljmpsiamsjam*
\begin{proof}
	We are in the following case.

	\[\begin{array}{cl}
			\statetwo=\infer{\tjudg{}{\tm\tmtwo}{\lintytwo}} 
			{\tjudg{}{\tm}{\arr{\mset{\ldots 
							\linty_i\ldots}}{\lintytwo}}
				& 
				p\;:=\;\tjudgi{}{\blue\tmtwo}{\ltyctxp{\initty^0_{\downpt}}_i(=\linty_i)}}
			& \tomachjump 
			
		\end{array}\]

		\[\begin{array}{lc}
			 \tomachjump &
			\jump{p} =\, \infer{\tjudg{}{\blue\var}{\ltyctxp{\initty^0_\downpt}_i (=\linty_i)}}{}
			=\state
		
		\end{array}\]

	Since $\statetwo$ is J-exhaustible there exists a judgement test $\statetwo_j:=\tjudg{}{\blue\tmtwo}{{\ltyctxp{\initty}_{i}}^\infty_\downpt,\ctxhole}$ such that $\rho':\statetwo_j\tomachbtone s_j\toksjam^*s_j'\tomachbttwo \state_j:=\tjudg{}{\blue\var}{{\ltyctxp{\initty}_{i}}^\infty_\downpt,\ctxhole}$. Now, $\rho'$ is a run on the Generalized \KSJAM.
	To get a \KSJAM run from $\statetwo$ to $\state$ we have to lift $\rho'$ using \reflemma{type-context-lifting}, obtaining the desired run 
	\[
	\rho\;:\;\tjudg{}{\blue\tmtwo}{\ltyctxp{\initty^\infty_{\downpt}}}\tomachbtone s\toksjam^*s'\tomachbttwo \tjudg{}{\blue\var}{\ltyctxp{\initty^\infty_{\downpt}}}
	\]
	We look now at the moreover part.
	By the J-exhaustible invariant, each state $s''_j=(\tyd, \ruleoc, (\ltyctxp{\initty}{},\ltyctxb), \pol,\infty)$ in $\sigma':s_j\toksjam^*s_j'$ has $\depth \ltyctxb\geq 1$. We call $\sigma$ the sub-run of $\rho$ such that $\sigma:s\toksjam^*s'$. It corresponds to the lifting of $\sigma'$, so each state $s''$ in $\sigma$ has the shape $s''=(\tyd, \ruleoc, (\ctxholep{\initty}, \ltyctxbp \ltyctx ), \pol,\infty)$. Notice that, since the \KSJAM is performing a $\tomachjump$, the state $q'$ has to have a backtracking depth equal to $0$. Therefore, by \refprop{simul-jmp-siam-ksjam}, we have $\depth{\ltyctx} = 2(\kparam - 0) + 1$.
	We can conclude noticing that for each state $s''$ in $\sigma$ we have $\depth{\ltyctxbp{\ltyctx}}=\depth \ltyctxb + \depth \ltyctx > 2\kparam +1$.	
\end{proof}
\section{Proofs of Section 6}
\label{sect:apx6}
\begin{restatable}{lemma}{lemmaantisubto}
\label{l:antisub-to}
	Let $\pi\pof \Gamma\vdash t\left\{x\leftarrow u\right\}:A$ then there exist an integer $n$ and some derivations $\pi'\pof \Gamma,x:[A_1,\dots,A_n]\vdash t:A$ and $\pi''_i\vdash u:A_i$ for $i\leq n$.

	Moreover $\maxto{\pi'}\leq \maxto \pi$ and $\maxto{\pi''_i}\leq \maxto{\pi}$ for each $i\leq n$.
\end{restatable}
\begin{proof}
	By induction on $\pi$.  In the axiom case we have two subcases.

	If $t=x$ then we have $\pi\pof\Gamma\vdash u:A$ and, since we are in a weak reduction setting, $u$ is closed so $\Gamma=\emptyset$.
	We then can set $n=1$ and $\pi''_1=\pi$, the  condition on $\pi''_1$ trivially holds.
	We also set $\pi'\pof \infer{\tjudg{\var:\mset{\linty}}{\var}{\linty}}{}$ and we obtain that $\maxto{\pi'}\leq \maxto \pi$ since $A$ appears already in $\pi$.

	If $t=y$ with $x\neq y$, than $\pi$ has the shape $\pi\pof\infer{\tjudg{y:\mset{\linty}}{y}{\linty}}{}$. We can then set $\pi'=\pi$ and $n=0$.\\

	The case $\tylamstar$ is similar. We take $\pi'$ as $\infer{\tjudg{}{t}{\star}}{}$ and $n=0$.\\

	For the $\tyapp$ case, the derivation $\pi$ has the following shape. 

	\[
		\infer{\pi\pof \Gamma \vdash sw \left\{x\leftarrow u\right\}: B}
		{\phi\pof \Gamma_1\vdash s\left\{x\leftarrow u\right\}:[B_1,\dots,B_k]\to B  &
		\mset{\gamma_i\pof \Delta_i\vdash w\left\{x\leftarrow u\right\}:B_i }_{i\in\mset{1,\ldots,k}}} 
	\]
	With $\Gamma=\Gamma_1\uplus\biguplus_{i\in\mset{1,\ldots,k}}\Delta_i$.
	We can apply the \ih on $\phi$ we get the existence of a $\phi'$, an integer $n_\phi$ and some derivations $\phi''_j$ for $j\leq n_\phi$.
	By applying the \ih on $\gamma_i$ for $i\leq k$ we get some $\gamma'_i$, some integers $n_{\gamma_i}$ and finally some derivations $\gamma''_{i,j}$ for $j\leq n_{\gamma_i}$.
	We can then set $n=n_\phi + \sum_{i\leq k}n_{\gamma_i}$ and take as $\pi''_i$ the union of all the $\phi''_j$ and $\gamma''_{i,j}$.
	We now have to check that the condition holds on every $\pi''_i$.
	By \ih we have that $\maxto {\phi''_j} \leq \maxto \phi$ for $j\leq n_\phi$ and $\maxto {\gamma''_{i,j}}\leq \maxto {\gamma_i}$ for $j\leq n_{\gamma_i}$ and $i\leq k$. 
	It is easy to see that $\maxto{\pi}=\max_{i\leq k}(\maxto \phi,\maxto{\gamma_i})$ and we can then conclude 
	\[
		\maxto{\pi''_i}\leq \max_{i\leq k}(\maxto{\phi},\maxto{\gamma_i})=\maxto \pi
	\]

	We now construct the derivation $\pi'$ as follows.
	\[
		\infer{\pi'\pof \Gamma,x:S\uplus \biguplus_{i\leq k}S_i \vdash sw: B}
		{\phi'\pof \Gamma_1,x:S\vdash s:[B_1,\dots,B_k]\to B  &
		\mset{\gamma'_i\pof \Delta_i,x:S_i\vdash w:B_i}_{i\in\mset{1,\ldots,k}}} 
	\] 
	By \ih we have that $\maxto{\phi'}\leq \maxto \phi$ and $\maxto{\gamma_i'}\leq \maxto{\gamma_i}$ for each $i\leq k$. So we can conclude with
	\[
		\maxto{\pi'}=\max_{i\leq k}(\maxto{\phi'},\maxto{\gamma'_i})\leq \max_{i\leq k}(\maxto{\phi},\maxto{\gamma_i})=\maxto{\pi}
	\] 

	We now check the case $\tylam$ with $\pi$ having the following shape.
	\[
		\infer{\pi\pof \Gamma\vdash \lambda y.s:[B_1,\dots,B_k]\to B}
		{\phi\pof \Gamma,y:[B_1,\dots,B_k] \vdash s: B } 
	\]
	By \ih on $\phi$ we get a $\phi'$, an integer $n_\phi$ and some $\phi''_j$ for $j\leq n_\phi$ such that $\maxto{\phi'}\leq \maxto \phi$ and $\maxto{\phi''_i}\leq \maxto \phi$. We set $n=n_\phi$ and $\pi''_i=\phi''_i$ and we have $\maxto{\pi''_i}=\maxto{\phi''_i}\leq \maxto{\phi}\leq \max_{i\leq k}(1+\sizeto B,\sizeto{B_i},\maxto{\phi})=\maxto{\pi}$. 

	We now contstuct $\pi'$ as follows.
	\[
		\infer{\pi'\pof \Gamma, x:[A_1,\dots,A_n]\vdash \lambda y.s:[B_1,\dots,B_k]\to B}
		{\phi'\pof \Gamma, x:[A_1,\dots,A_n],y:[B_1,\dots,B_k] \vdash s: B } 
	\]
	By \ih we have $\maxto{\phi'}\leq \maxto{\phi}$. Also $\maxto{\pi'}= \max(\sizeto{[B_1,\dots,B_k]\to B},\maxto{\phi'})$ and $\maxto{\pi}= \max(\sizeto{[B_1,\dots,B_k]\to B},\maxto{\phi})$. So $\maxto{\pi'}\leq \maxto{\pi}$.
\end{proof}

\begin{lemma}[Subject Expansion, Enhanced I]
\label{l:subj-exp-to}
	Let $t'$ be a closed $\lambda$-term, $\pi\pof \vdash t:A$ and $t' \towh t$. Then, there exists  $\pi'\pof \vdash t':A$ such that $\maxto{\pi'}\leq 1 + \maxto{\pi}$.  
\end{lemma}
\begin{proof}
	By induction on $t'\towh t$.\\

	In the base case we have $t'=(\lambda x.s)u\towh s\left\{x\leftarrow u\right\}=t$ and the derivation $\pi\pof\vdash s\left\{x\leftarrow u\right\}:\star$. Notice that $u$ is closed since $t'$ is closed. By \reflemma{antisub-to} we have some derivation $\phi'$ and $\phi''_i$ for $i\leq n$ such that $\maxto{\phi'}\leq \maxto{\pi}$ and $\maxto{\phi''_i}\leq \maxto{\pi}$. We can then construct the derivation $\pi'$ as in the following.
	\[
		\infer{\pi'\pof\vdash (\lambda x.s)u: A}
		{\infer{\vdash \lambda x.s:[A_1,\dots,A_n]\to A }{\phi'\pof x:[A_1,\dots,A_n]\vdash t: A}   &
		\mset{\phi_i''\pof\vdash u:A_i}_{i\in\mset{1,\ldots,n}} } 
	\]
	Notice that, by definition, $\sizeto{[A_1,\dots,A_n]\to A}=\max_{i\leq n}(1+\sizeto A, \sizeto{A_i})\leq \max_{i\leq n}(1+\maxto{\phi'},\maxto{\phi_i''})$.
	The maximal size of the new derivation $\pi'$ is 
	\begin{align*}
		\maxto{\pi'}=\max_{i\leq n}(\sizeto{A},\sizeto{\mset{A_1,\dots,A_n}\to A},\maxto{\phi'},\maxto{\phi''_i})\\
		\leq \max_{i\leq n}(\maxto{\phi'},1+\maxto{\phi'},\maxto{\phi''_i},\maxto{\phi'},\maxto{\phi''_i}) \\
		=\max_{i\leq n}(\maxto{\phi'}+1,\maxto{\phi_i''})\leq \max(\maxto{\pi}+1,\maxto{\pi})\leq \maxto{\pi}+1
	\end{align*}

	In the induction case we have $s'w\towh sw$. The derivation $\pi$ is as follows.
	\[
		\infer{\pi\pof \vdash sw: A}
		{\phi\;\pof\vdash s:[A_1,\dots,A_n]\to A   &
		\mset{\gamma_i\;\pof \vdash w:A_i}_{i\in\mset{1,\ldots,n}} } 
	\]
	By \ih on $\phi$ we get a derivation $\phi'$ such that $\maxto{\phi'}\leq \maxto{\phi}+1$. We then construct $\pi'$. 
	\[
		\infer{\pi'\pof \vdash s'w: A}
		{\phi'\;\pof\vdash s':[A_1,\dots,A_n]\to A   &
		\mset{\gamma_i\;\pof \vdash w:A_i}_{i\in\mset{1,\ldots,n}} } 
	\]
	Notice that $\maxto{\pi}=\max_{i\leq n}(\maxto{\phi},\maxto{\gamma_i})$.
	We obtain 
	\begin{align*}
	\maxto{\pi'}=\max_{i\leq n}(\maxto{\phi'},\maxto{\gamma_i})\leq \max_{i\leq n}(\maxto{\phi}+1,\maxto{\gamma_i})
	\\\leq \max_{i\leq n}(\maxto{\phi},\maxto{\gamma_i})+1
	=\maxto{\pi}+1
	\end{align*} 
\end{proof}

\propboundto*
\begin{proof}
	By induction on $n$. If $n=0$ then $t=\lambda x.u$ and $\pi\pof \infer{\tjudg{}{\lambda x.u}{\star}}{}$. In this case $\pi$ is a derivation made by only the rule $\tylamstar$, so $\maxto{\pi}=0$.

	Now let $t \towh t' \towh^{n-1} \lambda x.u$. By \ih it exists a $\pi'$ such that $\pi'\pof\vdash t':\star$ and $\maxto{\pi'}\leq n-1$, finally by \reflemma{subj-exp-to}, we obtain $\maxto{\pi}\leq\maxto{\pi'} +1\leq n$.
\end{proof}

\begin{restatable}{lemma}{lemmaantisubm}
	\label{l:antisub-m}
	Let $\pi\pof \Gamma\vdash t\left\{x\leftarrow u\right\}:A$ then there exist $\pi'\pof \Gamma,x:[A_1,\dots,A_n]\vdash t:A$, an integer $n$ and $\pi''_i\vdash u:A_i$ for $i\leq n$. Moreover 
	\begin{enumerate}
	\item $\maxmult{\pi'}\leq \maxmult \pi$
	\item $\maxmult{\pi''_i}\leq \maxmult{\pi}$ for each $i\leq n$
	\item $\size{\pi}=\size{\pi'}+ \sum_{i\leq n}\size{\pi''_i} - n$
	\item $n\leq \size \pi$
	\end{enumerate}  
\end{restatable}
\begin{proof}
	By induction on $\pi$.  In the axiom case we have two subcases.

	If $t=x$ then we have $\pi\pof\Gamma\vdash u:A$ and, since we are in a weak reduction setting, $u$ is closed so $\Gamma=\emptyset$. We then can set $n=1$ and $\pi''_1=\pi$, the  condition $(2)$ trivially holds. We also set $\pi'\pof \infer{\tjudg{\var:\mset{\linty}}{\var}{\linty}}{}$ and we obtain that $\maxto{\pi'}\leq \maxto \pi$ since $A$ appears already in $\pi$, so $(1)$ is proved.
	We have $n=\size{\pi'}=1$ and $\size \pi > 1$ since $u$ is closed and we are not in the $\tylamstar$ case. With this we get easely $(3)$ and $(4)$. 

	If $t=y$ and $x\neq y$, than $\pi$ has the shape $\pi\pof \infer{\tjudg{y:\mset{\linty}}{y}{\linty}}{}$. We can then set $\pi'=\pi$ and $n=0$.\\

	The case $\tylamstar$ is similar. We take $\pi'$ as $\infer{\tjudg{}{\lambda x.u}{\star}}{}$ and $n=0$.\\

	For the $\tyapp$ case the derivation $\pi$ has the following shape. 

	\[
		\infer{\pi\pof \Gamma \vdash sw \left\{x\leftarrow u\right\}: B}
		{\phi\pof \Gamma_1\vdash s\left\{x\leftarrow u\right\}:[B_1,\dots,B_k]\to B  &
		\mset{\gamma_i\pof \Delta_i\vdash w\left\{x\leftarrow u\right\}:B_i }_{i\in\mset{1,\ldots,k}}} 
	\]
	With $\Gamma=\Gamma_1\uplus\biguplus_{i\in\mset{1,\ldots,k}}\Delta_i$.
	We can apply the \ih on $\phi$ we get the existence of a $\phi'$, an integer $n_\phi$ and some derivations $\phi''_j$ for $j\leq n_\phi$.
	By applying the \ih on $\gamma_i$ for $i\leq k$ we get some $\gamma'_i$, some integers $n_{\gamma_i}$ and finally some derivations $\gamma''_{i,j}$ for $j\leq n_{\gamma_i}$.
	We can then set $n=n_\phi +\sum_{i\leq k}n_{\gamma_i}$ and take as $\pi''_i$ the union of all the $\phi''_j$ and $\gamma''_{i,j}$.
	We now have to check that the condition $(2)$ holds on every $\pi''_i$.
	By \ih we have that $\maxmult {\phi''_j} \leq \maxmult \phi$ for $j\leq n_\phi$ and $\maxmult {\gamma''_{i,j}}\leq \maxmult {\gamma_i}$ for $j\leq n_{\gamma_i}$ and $i\leq k$. 
	It is easy to see that for each $j\leq n$ 
	\[
	\maxmult{\pi''_j}\leq \max_{i\leq k}(\maxmult \phi,\maxmult{\gamma_i})=\maxmult{\pi}
	\]
	
	We now construct the derivation $\pi'$ as follows.
	\[
		\infer{\pi'\pof \Gamma,x:S\uplus \biguplus_{i\leq n}S_i \vdash sw: B}
		{\phi'\pof \Gamma_1,x:S\vdash s:[B_1,\dots,B_k]\to B  &
		\mset{\gamma'_i\pof \Delta_i,x:S_i\vdash w:B_i}_{i\in\mset{1,\ldots,k}}} 
	\] 
	By \ih we have that $\maxmult{\phi'}\leq \maxmult \phi$ and $\maxmult{\gamma_i'}\leq \maxmult{\gamma_i}$ for each $i\leq k$. It is easy to conclude $(4)$ with 
	\[
	\maxmult{\pi'}=\max_{i\leq k}(\maxmult{\phi'},\maxmult{\gamma'_i})\leq \max_{i\leq k}(\maxmult \phi,\maxmult{\gamma_i})=\maxmult{\pi}
	\]

	$(3)$ By \ih we have $\size \phi = \size{\phi'} + \sum_{j\leq n_\phi}\size{\phi''_j} - n_\phi$ and, for each $i\leq k$, $\size{\gamma_i} = \size{\gamma_i'} + \sum_{j\leq n_\phi}\size{\gamma''_{i,j}} - n_{\gamma_i}$. We obtain the following.
	\begin{align*}
		\size \pi = 1 + \size \phi + \sum_{i\leq k} \size{\gamma_i}
		\\=1 + \size{\phi'} + \sum_{j\leq n_\phi}\size{\phi''_j} - n_\phi + \sum_{i\leq k}(\size{\gamma_i'} + \sum_{j\leq n_\phi}\size{\gamma''_{i,j}} - n_{\gamma_i})\\
		= 1 + \size{\phi'}+\sum_{i\leq k}\size{\gamma'_i} + \sum_{j\leq n_\phi}\size{\phi''_j} + \sum_{i\leq k}\sum_{j\leq n_{\gamma_i}}\size{\gamma''_{i,j}} - (n_\phi + \sum_{i\leq k}n_{\gamma_i})
		\\ = \size{\pi'} + \sum_{i\leq n}\size{\pi''_i} - n
	\end{align*} 

	To prove $(4)$ we have that by \ih $n_\phi \leq \size \phi$ and for each $i\leq k$ we have $n_{\gamma_i}\leq \size{\gamma_i}$. So
	\[
		n=n_\phi +\sum_{i\leq k}n_{\gamma_i}\leq \size \phi + \sum_{i\leq k}\size{\gamma_i}+1=\size \pi
	\]

	We now check the case $\tylam$ with $\pi$ having the following shape.
	\[
		\infer{\pi\pof \Gamma\vdash \lambda y.s:[B_1,\dots,B_k]\to B}
		{\phi\pof \Gamma,y:[B_1,\dots,B_k] \vdash s: B } 
	\]
	By \ih on $\phi$ we get a $\phi'$, an integer $n_\phi$ and some $\phi''_j$ for $j\leq n_\phi$ such that $\maxmult{\phi'}\leq \maxmult \phi$ and $\maxmult{\phi''_i}\leq \maxmult \phi$. We set $n=n_\phi$ and $\pi''_i=\phi''_i$ and we have 
	\[
	\maxmult{\pi''_i}=\maxmult{\phi''_i}\leq \maxmult{\phi}\leq \max_{i\leq k}(k,\sizemult{B_i},\sizemult B,\maxmult{\phi})=\maxmult{\pi}
	\] 

	We now contstuct $\pi'$ as follows.
	\[
		\infer{\pi'\pof \Gamma, x:[A_1,\dots,A_n]\vdash \lambda y.s:[B_1,\dots,B_k]\to B}
		{\phi'\pof \Gamma, x:[A_1,\dots,A_n],y:[B_1,\dots,B_k] \vdash s: B } 
	\]
	By \ih we have $\maxmult{\phi'}\leq \maxmult{\phi}$ so we can conlude $(1)$ with the follwing.
	\begin{align*}
		\maxto{\pi'}= \max(\sizemult{[B_1,\dots,B_k]\to B},\maxmult{\phi'})
		\\\leq \max(\sizeto{[B_1,\dots,B_k]\to B},\maxto{\phi}) =\maxto{\pi}
	\end{align*}

	$(3)$ By \ih we have $\size{\phi} = \size{\phi'} + \sum_{j\leq n_\phi}\size{\phi''_j} - n_\phi$. So 
	\[
		\size{\pi} = 1 + \size{\phi} = 1 + \size{\phi'} + \sum_{j\leq n_\phi}\size{\phi''_j} - n_\phi = \size{\pi'} + \sum_{j\leq n}\size{\pi''_j} - n
	\]

	To prove $(4)$ we have that by \ih $n_\phi\leq \size \phi$ so $n=n_\phi\leq \size \phi+1=\size \pi$.

\end{proof}

\begin{lemma}[Subject Expansion, Enhanced II]
	\label{l:subj-exp-m}
	Let  $t$ be a closed term such that $\pi\pof \vdash t:\linty$,  and $\maxmult{\pi}\leq \size{\pi}$. If $t' \towh t$ then there exists  $\pi'\pof \vdash t':\linty$ such that $\maxmult{\pi'}\leq \size {\pi'}$.
\end{lemma}
\begin{proof}
	We will actually prove a stronger statement, showing that we also have $\size\pi\leq \size{\pi'}$. By induction on $t'\towh t$.\\

	In the base case we have $t'=(\lambda x.s)u\towh s\left\{x\leftarrow u\right\}=t$ and the derivation $\pi\pof\vdash s\left\{x\leftarrow u\right\}:\star$. Notice that $u$ is closed since $t'$ is closed. By \reflemma{antisub-m} we have some derivation $\phi'$ and $\phi''_i$ for $i\leq n$ such that $\maxmult{\phi'}\leq \maxmult{\pi}$ and $\maxmult{\phi''_i}\leq \maxmult{\pi}$. We can then construct the derivation $\pi'$ as in the following.
	\[
		\infer{\pi'\pof\vdash (\lambda x.s)u: A}
		{\infer{\vdash \lambda x.s:[A_1,\dots,A_n]\to A }{\phi'\pof x:[A_1,\dots,A_n]\vdash t: A}   &
		\mset{\phi_i''\pof\vdash u:A_i}_{i\in\mset{1,\ldots,n}} } 
	\]
	We have that $\size{\pi}=\size{\phi'}+\sum_{i\leq n}\size{\phi''_i}-n \leq \size{\phi'}+\sum_{i\leq n}\size{\phi''_i}+2=\size{\pi'}$.
	Notice that by definition $\sizemult{A}\leq \maxmult{\phi'}$ and $\sizemult{A_i}\leq \maxmult{\phi''_i}$ for each $i\leq n$ and by \reflemma{antisub-m} we have that $n\leq \size \pi$. With these inequalities we can prove the following.
	
	\begin{align*}
		\maxmult{\pi'}=max_{i\leq n}(\sizemult{[A_1,\dots,A_n]\to A},\maxmult{\phi'},\maxmult{\phi''_i})\\
		\leq \max_{i\leq n}(n,\sizemult{A_i}, \sizemult{A},\maxmult{\phi'},\maxmult{\phi''_i})\\\leq \max_{i\leq n}(\size \pi,\sizemult{A_i}, \sizemult{A},\maxmult{\phi'},\maxmult{\phi''_i}) 
		\\\leq
	\max_{i\leq n}(\size \pi,\maxmult{\phi'},\maxmult{\phi''_i})\leq
	\max(\size \pi, \maxmult \pi)=\size \pi \leq \size{\pi'}
	\end{align*} 
	In the induction case we have $s'w\towh sw$. The derivation $\pi$ is as follows.
	\[
		\infer{\pi\pof \vdash sw: A}
		{\phi\;\pof\vdash s:[A_1,\dots,A_n]\to A   &
		\mset{\gamma_i\;\pof \vdash w:A_i}_{i\in\mset{1,\ldots,n}} } 
	\]
	Notice that $\maxmult{\pi}=\max_{i\leq n}(\maxmult{\phi},\maxmult{\gamma_i})$ and obviously $\size \phi \leq \size \pi$.
	By \ih on $\phi$ we get a derivation $\phi'$ such that $\maxmult{\phi'}\leq \size \phi$ and $\size \phi \leq \size{\phi'}$. We then construct $\pi'$. 
	\[
		\infer{\pi'\pof \vdash s'w: A}
		{\phi'\;\pof\vdash s':[A_1,\dots,A_n]\to A   &
		\mset{\gamma_i\;\pof \vdash w:A_i}_{i\in\mset{1,\ldots,n}} } 
	\]
	By hypothesis we also have that $\maxmult{\pi}\leq \size\pi$. We can then prove the following inequalities.
	\begin{align*}
	\maxmult{\pi'}=\max_{i\leq n}(\maxmult{\phi'},\maxmult{\gamma_i})\leq 
	\max_{i\leq n}(\size \phi',\maxmult{\gamma_i})\leq \max(\size \phi, \maxmult{\pi})\\
	\leq \max(\size \pi, \maxmult{\pi}) = \size \pi \leq \size{\pi'}
	\end{align*}
	Finally, $\size \pi=\size \phi + \sum_{i\leq n}\size{\gamma_i}+1\leq \size{\phi'} + \sum_{i\leq n}\size{\gamma_i}+1=\size{\pi'}$.
\end{proof}

\propboundm*
\begin{proof}
	By induction on $n$. If $n=0$ then $t=\lambda x.u$ and we can construct a derivation $\pi\pof \infer{\tjudg{}{\lambda x.u}{\star}}{}$. In this case $\pi$ is a derivation made by only the rule $\tylamstar$, so $\maxmult{\pi}=0=\size \pi$.

	Now let $t \towh t' \towh^{n-1} \lambda x.u$. By \ih it exists a $\pi'$ such that $\pi'\pof t':\star$ and $\maxmult{\pi'}\leq \size{\pi'}$. We can then apply \reflemma{subj-exp-m}, obtaining a derivation $\pi$ such that $\pi\pof t:\star$ and $\maxmult{\pi}\leq \size\pi$.
\end{proof}

\lemmaboundtypesizek*
\begin{proof}
	 We start by noticing that for each type $B$ we have $\occstar{B}_0=1$. We now procede by induction on $A$. In the base case $\occstar{\star}_{1}=1$.\\

	For the inductive case we have $\linty=\mset{A_1,\dots,A_n}\to B$. We obtain the following
		\begin{align*}
			\occstar{\mset{A_1,\dots,A_n}\to B}_{k+1}=\sum_{i\leq n}\occstar{A_i}_{k} + \occstar{B}_{k+1}\\
			\leq \sum_{i\leq n}((\sizeto{A_i}\cdot \sizemult{A_i})^{k} + \occstar{A_i}_{k-1}) + (\sizeto{B}\cdot\sizemult{B})^{k+1} + \occstar{B}_{k}\\
			=\sum_{i\leq n}(\sizeto{A_i}\cdot \sizemult{A_i})^k + (\sizeto{B}\cdot\sizemult{B})^{k+1} +\sum_{i\leq n} \occstar{A_i}_{k-1} + \occstar{B}_{k}\\
			= \sum_{i\leq n}(\sizeto{A_i}\cdot \sizemult{A_i})^k + (\sizeto{B}\cdot\sizemult{B})^{k+1} +\occstar{\mset{A_1,\dots,A_n}\to B}_{k}
		\end{align*}
	So now we just need to prove $\sum_{i\leq n}(\sizeto{A_i}\cdot \sizemult{A_i})^k + (\sizeto{B}\cdot\sizemult{B})^{k+1}\leq (\sizeto{A}\cdot \sizemult{A})^{k+1}$. By definition we have that $\sizeto{\mset{A_1,\dots,A_n}\to B}=\max_{i\leq n}(1+\sizeto{B},\sizeto{A_i})$ so $\sizeto{B}\leq \sizeto{\mset{A_1,\dots,A_n}\to B} - 1=\sizeto{A}-1$. 
		\begin{align*}
			\sum_{i\leq n}(\sizeto{A_i}\cdot \sizemult{A_i})^k + (\sizeto{B}\cdot\sizemult{B})^{k+1}
			\\\leq\sum_{i\leq n}(\sizeto{A}\cdot \sizemult{A})^k + ((\sizeto{A}-1)\cdot\sizemult{A})^{k+1}\\
			\leq n(\sizeto{A}\cdot \sizemult{A})^k + ((\sizeto{A}-1)\cdot\sizemult{A})^{k+1}
			\\\leq \sizemult{A}(\sizeto{A}\cdot \sizemult{A})^k + ((\sizeto{A}-1)\cdot\sizemult{A})^{k+1} \\
			= \sizemult{A}^{k+1}(\sizeto{A}^k + (\sizeto{A}-1)^{k+1})\leq \sizeto{A}^{k+1} \cdot \sizemult{A}^{k+1}
		\end{align*}
	Where the last step comes from the fact that $a^k + (a-1)^{k+1}\leq a^{k+1}$ when $a>1$, as in this case since $\sizeto{A}=\sizeto{\mset{A_1,\dots,A_n}\to B}>1$.
\end{proof}


\thmboundpik*
\begin{proof}
	The previous corollary gives us a bound for every type $A$ occourring in $\pi$. To obtain the total weight $\WeightTimekJAM{\pi}{k}$, we multiply this bound by the number of judgements in $\pi$, \emph{\ie} by $\size{\pi}=\bigo{n^2}$, by \refprop{bound-size-beta}.
\end{proof}

\end{document}